\documentclass[11pt]{article}

\usepackage{amsthm,amsmath,amssymb}
\usepackage{eucal}
\usepackage{xifthen}
\usepackage{mathtools}
\usepackage{enumerate}
\usepackage{microtype}
\usepackage{xspace}
\usepackage{bm}
\usepackage{cite}
\usepackage{fancyhdr}
\usepackage{lastpage}
\usepackage[small]{caption}
\usepackage{xcolor}
\usepackage{algorithm}
\usepackage{algpseudocode}
\allowdisplaybreaks

\long\def\comment#1{}

\newfont{\bbb}{msbm10 scaled 700}

\newfont{\bb}{msbm10 scaled 1100}
\newcommand{\CC}{\mbox{\bb C}}

\newcommand{\RR}{\mbox{\bb R}}

\newcommand{\mbs}[1]{\bm{#1}}
\newcommand{\vect}[1]{{\lowercase{\mbs{#1}}}}

\newcommand{\Pmatrix}[1]{\begin{array}{ll}#1\end{array}}

\renewcommand{\Re}[1][]{\ifthenelse{\isempty{#1}}{\operatorname{Re}}{\operatorname{Re}\left(#1\right)}}
\renewcommand{\Im}[1][]{\ifthenelse{\isempty{#1}}{\operatorname{Im}}{\operatorname{Im}\left(#1\right)}}

\newcommand{\mv}{\vect{m}}

\newcommand{\Ac}{{\mathcal A}}

\newcommand{\Hc}{{\mathcal H}}

\newcommand{\Pc}{{\mathcal P}}

\newcommand{\Wc}{{\mathcal W}}

\newcommand{\Xc}{{\mathcal X}}
\newcommand{\Yc}{{\mathcal Y}}

\newcommand{\CN}[1][]{\ifthenelse{\isempty{#1}}{\mathcal{N}_{\mathbb{C}}}{\mathcal{N}_{\mathbb{C}}\left(#1\right)}}
\renewcommand{\P}[1][]{\ifthenelse{\isempty{#1}}{\mathbb{P}}{\mathbb{P}\left(#1\right)}}
\newcommand{\E}[1][]{\ifthenelse{\isempty{#1}}{\mathbb{E}}{\mathbb{E}\left(#1\right)}}
\renewcommand{\det}[1][]{\ifthenelse{\isempty{#1}}{\mathrm{det}}{\mathrm{det}\left(#1\right)}}
\newcommand{\trace}[1][]{\ifthenelse{\isempty{#1}}{\mathrm{tr}}{\mathrm{tr}\left(#1\right)}}
\newcommand{\rank}[1][]{\ifthenelse{\isempty{#1}}{\mathrm{rank}}{\mathrm{rank}\left(#1\right)}}
\newcommand{\diag}[1][]{\ifthenelse{\isempty{#1}}{\mathrm{diag}}{\mathrm{diag}\left(#1\right)}}

\DeclarePairedDelimiter\Abs{\lvert}{\rvert^2}

\renewcommand{\arg}{{\hbox{arg}}}

\renewcommand{\Re}{{\rm Re}}
\renewcommand{\Im}{{\rm Im}}

\allowdisplaybreaks
\newcommand{\st}{{\rm s.t.}}
\DeclareMathAlphabet{\mathcal}{OMS}{cmsy}{m}{n}

\newcommand{\defeq}{\triangleq}

\newtheorem{remark}{Remark}
\newtheorem{definition}{Definition}
\newtheorem{theorem}{Theorem}
\newtheorem{example}{Example}
\newtheorem{lemma}{Lemma}

\usepackage{hyperref}
\hypersetup{
    bookmarks=true,         
    unicode=false,          
    pdftoolbar=true,        
    pdfmenubar=true,        
    pdffitwindow=false,     
    pdfstartview={FitH},    
    pdfnewwindow=true,      
    colorlinks=true,       
    linkcolor=red,          
    citecolor=green,        
    filecolor=magenta,      
    urlcolor=cyan           
}

\newcommand\blfootnote[1]{%
  \begingroup
  \renewcommand\thefootnote{}\footnote{#1}%
  \addtocounter{footnote}{-1}%
  \endgroup
}

\begin{document}
\date{}
\title{Opportunistic Treating Interference as Noise}
\author{ \normalsize Xinping Yi and Hua Sun\\
}
\maketitle

\blfootnote{ 
Xinping Yi (email: xinping.yi@liverpool.ac.uk) is with the Department of Electrical Engineering \& Electronics at University of Liverpool, UK.
Hua Sun (email: hua.sun@unt.edu) is with the Department of Electrical Engineering at University of North Texas, USA. 
}

\begin{abstract}
We consider a $K$-user interference network with $M$ states, where each transmitter has $M$ messages and over State $m$, Receiver $k$ wishes to decode the first $\pi_k(m) \in \{1,2,\cdots,M\}$ messages from its desired transmitter.
This problem of channel with states models opportunistic communications, where more messages are sent for better channel states. The first message from each transmitter has the highest priority as it is required to be decoded regardless of the state of the receiver; the second message is opportunistically decoded if the state allows a receiver to decode 2 messages; and the $M$-th message has the lowest priority as it is decoded if and only if the receiver wishes to decode all $M$ messages.
For this interference network with states, we show that if any possible combination of the channel states satisfies a condition under which power control and treating interference as noise (TIN) are sufficient to achieve the entire generalized degrees of freedom (GDoF) region of this channel state by itself, then a simple layered superposition encoding scheme with power control and a successive decoding scheme with TIN achieves the entire GDoF region of the network with $M$ states for all $KM$ messages.

\end{abstract}
\newpage

\section{Introduction}
\label{sec:intro}
Opportunistic communication refers to the opportunistic utilization of channel resources and the adaptation to network dynamics for efficient data transmission.
The early study in this regard dates back to downlink multiuser scheduling in time-varying wireless channels \cite{wc_book_tse,Pramod}. By opportunistically beamforming towards the user with the best channel, the base station exploits the multiuser diversity gain \cite{wc_book_tse} so as to maximize the overall system throughput \cite{Pramod}.
A similar idea has also been explored in cognitive radio systems for dynamic spectrum management \cite{devroye2008cognitive,goldsmith2009breaking}, in which the secondary users are assisted to access the spectrum licensed to the primary users opportunistically, in order to ensure efficient communication of secondary users without worsening the performance of primary users.

While existing opportunistic communication techniques are mainly placed at the transmitter side, the focus of this work is on opportunistic decoding at the receiver side, exploiting the benefits of varying decoding capabilities in dynamic networks.
When the channel condition is better, we wish to take this advantage and achieve a higher communication rate, while if the channel condition turns out to be bad, we will lower the expectation but a certain basic communication rate is still guaranteed. From the information theoretic perspective, this problem is typically modeled as communicating several message sets over a channel with states, where the base message set (corresponding to the basic communicate rate) must be transmitted successfully regardless of the state, and the opportunistic message set (corresponding to the higher communication rate) will also go through for a better channel state. Such formulations have been previously studied in the context of a single user slow fading channel with multiple antennas from a outage probability perspective (diversity-multiplexing tradeoff) \cite{Diggavi_Tse} and a two user bursty interference channel (where interference is not present for the better channel state) from an approximate capacity perspective \cite{Khude_Prabhakaran_Viswanath, Khude_Prabhakaran_Viswanath2}.

In this work, we go beyond two users and consider a general $K$-user Gaussian interference network, albeit with specific restrictions on the channel strength. In particular, we are interested in a broad regime where the simple and practical strategy of treating interference as noise (TIN) has been shown to be approximately optimal in the sense that the generalized degrees of freedom (GDoF) region is achieved by TIN \cite{Geng_TIN}. The optimality of TIN has since been explored beyond the regular interference channel, to $X$ message sets \cite{TIN-X} (where each transmitter has a message for each receiver), to the parallel channel setting \cite{Sun_Jafar_ParallelTIN, Sun_Jafar_ParallelTINRegion} (where each user pair is connected by a number of parallel channels), to the compound channel setting \cite{TIN-Compound} (where there is only one message for each user pair and the message must be reliably decoded regardless of the realization of the compound state), and to the interfering multiple access channel setting \cite{Aydin-TIN,TIN-IMAC} (where each receiver has one more paired transmitter carrying independent messages). Besides the characterization of GDoF regions, another important problem on power control has been considered in \cite{TIN-X,Yi-TIN}, where a number of low-complexity power allocation algorithms were proposed. Inspired by the TIN optimality conditions, efficient distributed link scheduling mechanisms were proposed in \cite{ITLinQ,ITLinQ+} for spectrum sharing in device-to-device communications, demonstrating an interesting translation from theory to practice.

\subsection{Motivating Example}
We are inspired by the observation that TIN naturally fits the opportunistic communication scenario, illustrated in the following example. Consider a 3 user interference channel with 2 states, as shown in Figure \ref{fig:ex}. In the first state, the network is fully connected and the channel strength for each link is depicted (the channel strength is measured in dB scale. For a detailed explanation, refer to the system model section). In the second state, each receiver only sees one interfering transmitter (due to, say, time-varying channel statistics), i.e., the red dashed interfering links are not present (e.g., Receiver 1 is interfered only by Transmitter 2, but not by Transmitter 3). Both states are in the regime where TIN is optimal \cite{Geng_TIN}. 

\begin{figure}[h]
\center
\includegraphics[width=5.5 in]{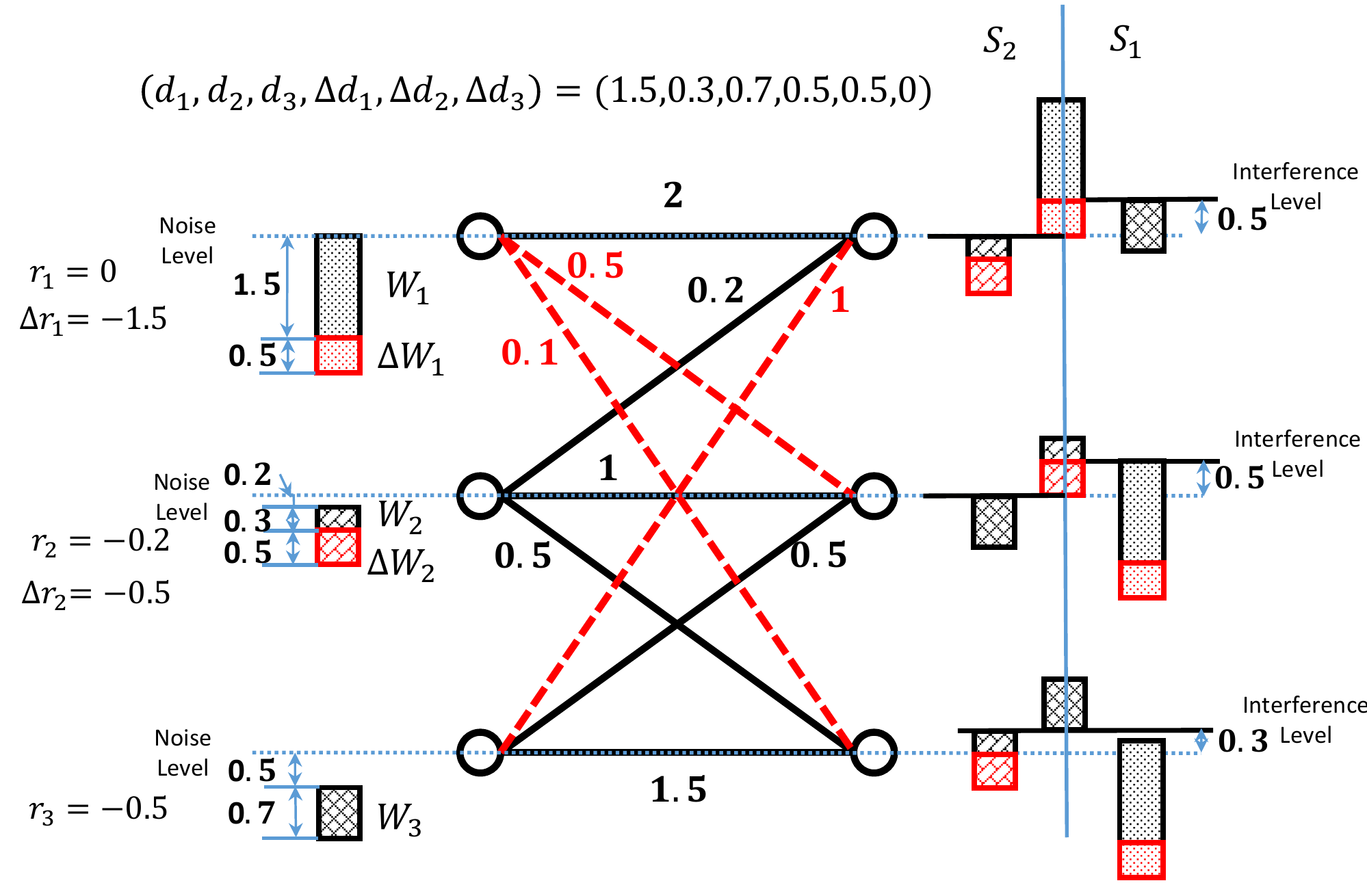}
\caption{\small A 3 user interference network with 2 states where the dashed red interfering links are not present in the second (better) state. Over the second state, the opportunistic message set $(\Delta W_1, \Delta W_2, \Delta W_3)$ is sent in addition to the base message set $(W_1, W_2, W_3)$. The transmitted power levels of the messages and the interference power levels are shown. At the receiver side, to the right of the blue vertical line (labelled as $S_1$), we have the interference power level for the fully connected state and to the left (labelled as $S_2$), we have the interference power level for the partially connected state. The exposed signal levels are exploited to send the opportunistic message set.}
\label{fig:ex}
\end{figure}

We wish to send 3 messages $(W_1, W_2, W_3)$ ($W_i$ for the $i$-th user pair) over the first state and the associated GDoF tuple for the messages is $(d_1, d_2, d_3) = (1.5, 0.3, 0.7)$.
A TIN scheme that achieves this GDoF tuple is shown in Figure \ref{fig:ex}, where the transmit power levels and the received interference power levels are explicitly shown (the power levels are measured in dB scale as well). For example, $W_2$ is sent at power level $-0.2$ so that it is received at Receiver 2 at power level $-0.2 + 1 = 0.8$ (where 1 is the channel strength from Transmitter 2 to Receiver 2) and it is received at Receiver 1 at power level $-0.2 + 0.2 = 0$. From Figure \ref{fig:ex}, it is easy to verify that the desired GDoF value is achieved at each receiver by TIN (the interference power level is lower than that of the desired message by the exact amount of the GDoF value).

Next we consider the performance of the same scheme over the second state (the better state with less interference). We notice that because some interfering links become missing, some signal levels that were occupied by interference are left interference-freely. For example, consider Receiver 2, where previously the interference power level was 0.5 (caused by Transmitter 1). Now as the interfering link from Transmitter 1 is not present, the interference power level drops to 0 (caused by Transmitter 3). In other words, the signal level from power 0 to 0.5 is now clean and we may naturally use this signal level to send the opportunistic message $\Delta W_2$ to achieve the GDoF value of $\Delta d_2 = 0.5$ (see the red tilted rectangle in Figure \ref{fig:ex}). Note that this will not influence the base message set as the exposed signal level is always lower than that of the base message set and the opportunistic message will not increase the interference power level at undesired receivers. Similarly, Transmitter 1 will send $\Delta W_1$ with the exposed signal level to achieve $\Delta d_1 = 0.3$ (see the red dotted rectangle in Figure \ref{fig:ex}). While for Receiver 3, its interference level is not decreased even if the interfering link from Transmitter 1 disappears, leaving no room for $\Delta W_3$ so that the opportunistic message for Transmitter 3 will not be sent. To decode the opportunistic message, each receiver first decodes the base message and then successively proceeds to decode the opportunistic message, both by TIN. To summarize, we have achieved the GDoF tuple of $(\Delta d_1, \Delta d_2, \Delta d_3) = (0.3, 0.5, 0)$ opportunistically.

From this example, we see that the key idea of our achievable scheme is to superpose the opportunistic message set over the base message set, using the largest power that is not exploited yet, to fulfill the interference-free signal level that is opportunistically exposed due to the decrease of interference strength. We may vary the power levels and the GDoF tuple for the base messages arbitrarily.
A natural question is: is this scheme - superposition encoding with power control and successive decoding with TIN - information theoretically optimal? We answer this question in the affirmative in this paper and explore the general channel conditions under which the proposed scheme is optimal.

\subsection{Main Contribution}
Interestingly, the natural scheme of superposition encoding and successive decoding with TIN is information theoretically optimal for a broad set of channel conditions and a broad class of message setting. Specifically, we consider a $K$-user interference network with $M$ states, where each transmitter has $M$ messages ordered by their importance (where the first message is the most important and the $M$-th message is the least important), and each receiver will decode the first $\pi \in \{1,2,\cdots,M\}$ messages ($\pi$ might differ across channel states and across receivers).

As the main result of this work, we show that if all subnetworks (given by the $K$ transmitters and $K$ receivers from possibly different states) of the $K$-user interference network satisfy the TIN-optimality condition identified in \cite{Geng_TIN}, then for arbitrary realizations of $\pi$ (arbitrary decoding thresholds across the states and the receivers), the simple scheme of layered superposition coding with TIN achieves the entire GDoF region.

\bigskip
We begin by defining the notations.

\underline{Notations}: For an integer $N$, we define $[N] \triangleq \{1,2,\ldots,N\}$. Given $n \in [N]$, we denote by $\{a(n)\}_n$ a set of $a(n)$ with all $n$, i.e., $\{a(n)\}_n \defeq \{a(1), a(2), \dots, a(N)\}$, and similarly $\{a(m,n)\}_{m,n}$ given $m \in [M]$ and $n \in [N]$ is a set with $MN$ elements, i.e., $\{a(m,n)\}_{m,n} \defeq \{a(1,1),  a(1,2),  \dots,  a(1,N),$  $a(2,1),  \dots,  a(M,N)\}$. {We also denote by $a([n_1:n_2])$ a subset of $a(n)$ with $n_1 \le n \le n_2$, i.e., $a([n_1:n_2]) \defeq \{a(n_1), a(n_1+1), \dots, a(n_2)\}$.}

\section{System Model}\label{sec:systemmodel}
\subsection{Gaussian Interference Network with States}
Consider the $K$-user single-antenna Gaussian interference network with $M$ states. The received signal for Receiver $k$ over the $t$-th channel use when the network falls into the $m$-th state is given by
\begin{equation}
\label{original}
Y_k^{[m]}(t) = \sum_{i=1}^{K}{h}_{ki}^{[m]} \tilde{X}_i(t) + {Z}_k^{[m]}(t),~~~\forall k \in [K], \forall m \in [M]
\end{equation}
where ${h}_{ki}^{[m]}$ is the channel coefficient from Transmitter $i$ to Receiver $k$ at the $m$-th state. The $K^2$-ary channel coefficients tuple at the $m$-th state $(\{h_{ki}^{[m]}\}_{k,i})$ is taken from a finite set $\Hc$, and is fixed within each state but can vary across states. The additive white Gaussian noise (AWGN) for Receiver $k$ over the $t$-th channel use ${Z}_k^{[m]}(t)$ has zero mean and unit-variance. The AWGN processes at all receivers are i.i.d over time.

The set of channel coefficients $\Hc$ over all $M$ states is available at all transmitters and receivers. Over different states, a possibly different set of messages is required to be communicated reliably (as detailed below). An interpretation\footnote{Equivalently, this channel model with states represents a multicast scenario where each state has a different set of $K$ receivers and the receivers across different states have different decoding requirements.} of this channel model with states is that the $M$ states represent the channel uncertainty at the transmitters. The transmitters know that the channels could be in any one of the $M$ states, but otherwise has no knowledge about which state the network falls into exactly. However, the transmitters wish to communicate opportunistically, i.e., if the network turns out to be in a better state, more messages are communicated. The receiver is aware of the exact state of the network and depending on the state, he will choose which set of messages to decode. A detailed description of the encoding and decoding operations is as follows.

\underline{Encoding}: Each Transmitter $i$ has a set of independent messages $\{W_i^{[m]}\}_{m=1}^M$, each of which is uniformly distributed over the message index set $\Wc_i^{[m]} \defeq \{1,2,\ldots,\lceil 2^{nR_i^{[m]}}\rceil\}$. These messages are jointly mapped to the codeword $ \{\tilde{X}_i(t)\}_{t=1}^n$ (abbreviated as $\tilde{X}_i^n \in \Xc_i^n$) that is transmitted over $n$ channel uses, and is subject to the average power constraint, $\sum_{t=1}^{n} \mathbb{E}\left [| \tilde{X}_i(t) |^2 \right] \leq n P_i$ where the expectation is over all the candidate messages. The message-to-codeword mapping for Transmitter $i$ ($i \in [K]$) is described by the following encoding function, 
\begin{align}
f_i: \textstyle \prod_{m=1}^M \Wc_i^{[m]} \mapsto \Xc_i^n .
\end{align}
Note that a single encoding mapping is used at each transmitter.

\underline{Decoding}: Suppose the channels are at the $m'$-th state.  For Receiver $k$, the received signal $ \{Y_k^{[m']}(t)\}_{t=1}^n$ (abbreviated as $Y_{k,m'}^{n} \in \Yc_k^{[m']}$)  is used to produce the estimates $\{\hat{W}_k^{[m]}\}_{m=1}^{\pi_k(m')}$ of the messages $ \{{W}_k^{[m]}\}_{m=1}^{\pi_k(m')}$. 
Among these messages, ${W}_k^{[1]}$ is referred to as the basic message that must be decoded at any state, and $\big\{{W}_k^{[m]}, m \in \{2,\dots,\pi_k(m')\}\big\}$ are the additional messages to be opportunistically decoded, referred to as ``opportunistic messages''. The total number of messages $\pi_k(m')$ to be decoded by Receiver $k$ at the $m'$-th state is fixed and globally known. 
$\pi_k(m')$ can be any number in $[M]$ so that $(\pi_k(1), \cdots, \pi_k(M)) \in [M]^M$. 
In other words, at State $m'$, from $Y_{k,m'}^{n}$ we need to decode messages ${W}_k^{[1]}, {W}_k^{[2]}, \cdots, {W}_k^{[\pi_k(m')]}$.
Thus, the decoding function at the $m'$-th state is given by
\begin{align}
g_k^{[m']}: \Yc_k^{[m']} \mapsto  \textstyle \prod_{m=1}^{\pi_k(m')} \Wc_k^{[m]}, \quad \forall m' \in [M].
\end{align}
{Note that the decoding functions $g_k^{[m']}$ can be distinct for different states $m'$. Fig. \ref{fig:decoding-fun} gives an example of a 3-user network with 3 states, where for Receiver 1 $\pi_1(1)=2,\pi_1(2)=1, \pi_1(3)=3$, for Receiver 2 $\pi_2(1)=1,\pi_2(2)=2, \pi_2(3)=3$, and for Receiver 3 $\pi_3(1)=3,\pi_3(2)=2, \pi_3(3)=1$. For each receiver, the basic message is always decodable at all states and in this case $(\pi_k(1),\pi_k(2), \pi_k(3) )$ is a permutation of $\{1,2,3\}$. 
\begin{figure}[h]
\center
\includegraphics[width=4.5 in]{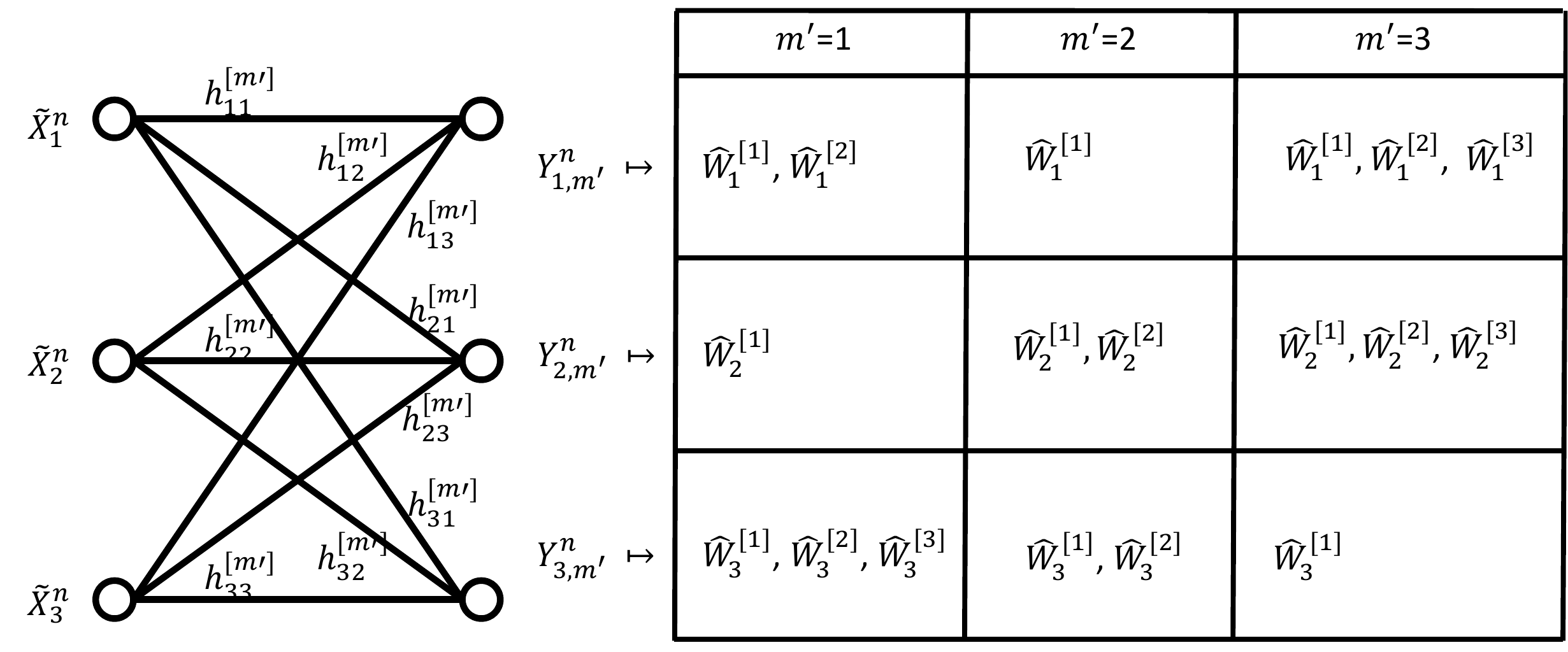}
\caption{\small A 3-user interference network with 3 states. Each transmitter $i$ has $3$ messages, $W_i^{[1]}, W_i^{[2]}, W_i^{[3]}$, to send. Over the first state, Receiver 1 needs to decode $W_1^{[1]}, W_1^{[2]}$, Receiver 2 needs to decode $W_2^{[1]}$, and Receiver 3 needs to decode $W_3^{[1]}, W_3^{[2]}, W_3^{[3]}$. The messages that each receiver needs to decode over the two remaining states are shown in the figure.}
\label{fig:decoding-fun}
\end{figure}

The average probability of error is defined as follows
\begin{align}
P_e^{(n)} = \Pr \left( \bigcup_{m'=1}^M \left\{ \big(\{W_{k}^{[1:\pi_k(m')]} \}_{k} \big) \neq \big(\{\hat{W}_{k}^{[1:\pi_k(m')]}\}_{k} \big) \right\} \right),
\end{align}
where we take the union of all $M$ states because decoding error of any state will result in an error event (i.e., we need to maintain reliable communication over all states), and at each state the error events of all basic and opportunistic messages for this state across all users are counted.

A rate tuple $(\{R_k^{[m]}\}_{k,m})$ is said to be achievable if we have a set of encoding $\{f_i\}_i$ and decoding functions $\{g_k^{[m]}\}_{k,m}$ such that $P_e^{(n)} \to 0$ as $n \to \infty$. The capacity region $\mathcal{C}$ is the closure of the set of all achievable rate tuples. 

\subsection{GDoF Framework} \label{sec:gdof-framework}
Following  \cite{Geng_TIN}, we now translate the channel model (\ref{original}) into an equivalent normalized form to facilitate GDoF studies. For such a purpose, we define $\tilde{X}_i(t) = \sqrt{P_i}{X}_i(t)$. Then over the $t$-th channel use, the received signal for Receiver $k$ at the $m$-th state is described by
\begin{align}
\label{now}
Y_k^{[m]}(t) &= \sum_{i=1}^{K}  h_{ki}^{[m]}  \sqrt{P_i} {X}_i(t) + {Z}_k^{[m]}(t) \\
&= \sum_{i=1}^{K} \sqrt{P^{\alpha_{ki}^{[m]}}} e^{j \theta_{ki}^{[m]}}  {X}_i(t) + {Z}_k^{[m]}(t)  \label{use}
\end{align}
where we take $P >1$ as a nominal power value, and define\footnote{As noted in \cite{Geng_TIN},  avoiding negative $\alpha$'s, will not influence the GDoF results.}
$\alpha_{ki}^{[m]} \triangleq \left( {\log \left(  \Abs{h_{ki}^{[m]} } P_i \right)} / {\log P} \right)^+.$
Now the power constraint becomes $\frac{1}{n}\sum_{t=1}^{n} \mathbb{E} \left[ | {X}_i(t) |^2 \right] \leq 1.$
As in \cite{Geng_TIN}, we call $\alpha_{ki}^{[m]}$ the channel strength level (exponent). The equivalent model (\ref{use}) will be used in the rest of this paper.

Next, we introduce the encoding function used in this work.

\begin{definition}[Simple Layered Superposition Coding] 
In simple layered superposition coding, the transmitted signal is produced by 
\begin{align}
X_i(t) = \sum_{m=1}^M X_i^{[m]}(t), \quad \forall i
\end{align}
where each message $W_i^{[m]}$ is separately encoded by an independent Gaussian codebook $\{X_i^{[m]}(t)\}_t$ with power $P^{r_i^{[m]}}$, i.e., $\E {[\Abs{X_i^{[m]}(t)}]}=P^{r_i^{[m]}}$ and then the codewords are added (superposed).
Further, we assume that the power decreases with the order of the message $m$,\footnote{It is worthy noting that the message order is not the same as the state index. Receiver $k$ at the $m$-th state is able to decode messages with order up to $\pi_k(m)$.} i.e., $0 \ge r_i^{[1]} \ge r_i^{[2]} \ge \dots \ge r_i^{[M]}$.\end{definition}

The encoded messages are superposed in a layered manner according to the power. For a power layer illustration, we put the basic message at the top layer, followed successively by the opportunistic messages of next orders, and the opportunistic message of order $M$ is layered at the bottom. Note that the above power allocation must satisfy the sum power constraint 
$\sum_{m=1}^M  P^{r_i^{[m]}} \le 1, \forall i$.

In this work, we consider the TIN setting and use a single set of decoding functions for all $\{g_{k}^{[m']}\}_{m'}$ (with parameters varying to conduct opportunistic decoding). We refer to this class of decoding functions as ``Opportunistic TIN'', defined as follows.

\begin{definition}[Opportunistic TIN] At the receiver side, opportunistic TIN is a successive interference cancelation based decoding rule where opportunistically the interference is treated as Gaussian noise. The basic message is first decoded while treating the interference caused by all opportunistic messages as Gaussian noise. As a sequel, the corresponding signal carrying the basic message can be reconstructed using the known channel state information at the receivers and then subtracted from the received signal. The residual received signal can be successively used to recover the lower layer opportunistic messages. Such a decoding-reconstructing-subtracting procedure repeats until the opportunistic messages of interest at the present state are successively recovered.
\end{definition}

Let us consider State $m'$, where Receiver $k$ is interested in decoding messages $\{W_k^{[m]}\}_{m=1}^{\pi_k(m')}$ while treating the remaining opportunistic messages $\{W_k^{[m]}\}_{m=\pi_k(m')+1}^{M}$ as noise.
The received signal at the $m'$-th state for Receiver $k$ can be rewritten as
\begin{align}
Y_k^{[m']}(t) &= \sum_{i=1}^{K} \sum_{m=1}^M  \sqrt{P^{\alpha_{ki}^{[m']}}} e^{j \theta_{ki}^{[m']}}  {X}_i^{[m]}(t) + {Z}_k^{[m']}(t).
\end{align}
The successive interference cancellation starts with the basic message $W_k^{[1]}$ where the interference from all opportunistic messages is treated as noise. After $W_k^{[1]}$ is decoded, the signal $X_k^{[1]}$ is reconstructed and subtracted from the received signal.
After applying $m-1, m \le \pi_k(m')$ rounds of successive interference cancellation, the messages $\{W_k^{[1]},\dots,W_k^{[m-1]}\}$ are successively decoded and the corresponding signals are subsequently subtracted. At the $m$-th round, the residual received signal can be written as
\begin{align}
\overline{Y}_k^{[m]}(t)  &= Y_k^{[m']}(t) -  \sum_{m''=1}^{m-1}  \sqrt{P^{\alpha_{kk}^{[m']}}} e^{j \theta_{kk}^{[m']}}  {X}_k^{[m'']}(t)   \\
&= \sqrt{P^{\alpha_{kk}^{[m']}}} e^{j \theta_{kk}^{[m']}}  {X}_k^{[m]}(t)
+ \sum_{m''=m+1}^{M}  \sqrt{P^{\alpha_{kk}^{[m']}}} e^{j \theta_{kk}^{[m']}}  {X}_k^{[m'']}(t) \\
& \qquad + \sum_{i=1, i \ne k}^{K} \sum_{m''=1}^M  \sqrt{P^{\alpha_{ki}^{[m']}}} e^{j \theta_{ki}^{[m']}}  {X}_i^{[m'']}(t) + {Z}_k^{[m']}(t).
\end{align}
Thus, the signal-to-interference-and-noise (SINR) ratio for the desired signal $X_k^{[m]}(t)$ is
\begin{align}
{\rm SINR}_{k}^{[m]} (m') &= \frac{P^{\alpha_{kk}^{[m']}+r_k^{[m]} }}{1 + \sum_{m'' \in [m+1:M]} P^{\alpha_{kk}^{[m']}+r_k^{[m'']} } + \sum_{i: i \neq k } \sum_{m'' \in [M]} P^{\alpha_{ki}^{[m']}}  P^{r_i^{[m'']}} }.
\end{align}
Then the achievable rate of $W_k^{[m]}$ is given by
\begin{align} \label{eq:min-rate}
R_k^{[m]} &= \min_{m' : \; \pi_k(m') \geq m} \left\{ \log \big(1+ {\rm SINR}_k^{[m]}(m') \big) \right\}\\
& =\log \Big(1+ \min_{m': \; \pi_k(m') \geq m} \big\{ {\rm SINR}_k^{[m]}(m') \big\} \Big)
\end{align}
where the $\min$ operation is to make sure $W_k^{[m]}$ can be reliably decoded at all states that are supposed to decode no less than $m$ messages, i.e., for all $m' \in [M]$ such that $\pi_k(m') \geq m$.
Therefore the GDoF $d_k^{[m]} = \lim_{P \rightarrow \infty} \frac{R_k^{[m]}}{\log P}$ is given by
\begin{align}
d_k^{[m]} &=  \max \left\{0, \min_{m': \; \pi_k(m') \geq m} \Big\{\alpha_{kk}^{[m']} + r_k^{[m]} - \max \big\{0, \alpha_{kk}^{[m']} + r_k^{[m+1]}, \max_{i: i \neq k} (\alpha_{ki}^{[m']} + r_i^{[1]}) \big\} \Big\} \right\}\\
&= \max\left\{0, \min \Big\{r_k^{[m]} - r_k^{[m+1]}, \min_{m': \; \pi_k(m') \geq m} \big\{\alpha_{kk}^{[m']} + r_k^{[m]} - \max \{0, \max_{i: i \neq k} (\alpha_{ki}^{[m']} + r_i^{[1]})\} \big \} \Big\} \right\} \label{eq:GDoF}
\end{align}
where the last step follows from the fact that $r_k^{[m]}$ is decreasing in $m$. 
{Note that $r^{[M+1]}_k, \forall k$ is an auxiliary power variable introduced to simplify the GDoF expression and it is convenient to interpret $r^{[M+1]}_k$ as a negative number in the range of $(-\infty , r_k^{[M]}]$ that represents the lowest power level used by the messages.
}

We define the GDoF region as
\begin{eqnarray}
\mathcal{D} \triangleq \left\{ (\{d_k^{[m]}\}_{k,m}) : d_k^{[m]} = \lim_{P \rightarrow \infty} \frac{R_k^{[m]}}{\log P}, \forall k \in [K], m \in [M], (\{R_k^{[m]}\}_{k,m}) \in \mathcal{C} \right\}.
\end{eqnarray}

\section{Main Result}
The main result of this work, stated in the following theorem, is that simple layered superposition coding and opportunistic TIN decoding is GDoF optimal under a broad set of channel conditions.

\begin{theorem}\label{thm:main}
Consider an $M$-state $K$-user single-antenna Gaussian interference channel with channel strength exponents $\{\alpha_{ij}^{[m]}\}_{i,j,m}$.
If the following condition
\begin{align} \label{eq:tin-optimal}
\alpha_{kk}^{[m_{k}]} \ge &\max_{j: j \neq k} \{ \alpha_{jk}^{[m_j]}\} + \max_{i: i \neq k} \{\alpha_{ki}^{[m_k]}\}, \nonumber \\
& \forall i, \; j, \; k \in [K], \quad \forall m_j,m_k \in [M]
\end{align}
is satisfied, then power control with simple layered superposition coding at the transmitters and opportunistic TIN at the receivers achieves the entire GDoF region, which includes all GDoF tuples $(\{d_k^{[m]}\}_{k,m}) \in \RR_+^{MK}$ satisfying
\begin{subequations} \label{eq:GDoF-region}
\begin{align}
\sum_{m=1}^{\pi_k(m')}  d_k^{[m]}  &\le \alpha_{kk}^{[m']} , \quad \forall m' \in [M], \; \forall k \in [K] \label{eq:e1}\\
 \sum_{k=1}^{k'}  \sum_{m=1}^{\pi_{i_k}(m_{i_k})} d_{i_k}^{[m]} &\le  \sum_{k=1}^{k'} (\alpha_{i_k i_k}^{[m_{i_k}]} - \alpha_{i_k i_{k+1}}^{[m_{i_k}]}), \label{eq:e2} \\
  \forall  (i_1, i_2, \dots, i_{k'}) \in \Pi_{k'},& \; \forall (m_{i_1}, m_{i_2}, \dots, m_{i_{k'}}) \in [M]^{k'}, \; \forall k' \in [K] \backslash \{1\},  
 \end{align}
 \end{subequations}
 where $\Pi_{k} \subseteq [K]$ is the collection of all possible cyclically ordered $k$-element subsets of user indicies without repetition, e.g., $\Pi_2=\{(1,2),(1,3),(2,3)\}$ and $\Pi_3=\{(1,2,3),(1,3,2)\}$, and $[M]^{k'}$ is a set with cardinality $M^{k'}$ collecting all possible $k'$-ary tuples, in which each coordinate is from $[M]$, e.g., $[2]^{3}=\{(1,1,1), (1,1,2), (1,2,1), (1,2,2), (2,1,1), (2,1,2), (2,2,1), (2,2,2)\}$. 
The number of messages decoded by Receiver $k$ at the $m$-th state  $\pi_k(m)$ is arbitrarily chosen from $[M]$, and is globally known {a priori}.
\end{theorem}

\begin{remark} \normalfont
\label{remark:remark1}
The TIN optimality condition \eqref{eq:tin-optimal} and the GDoF region \eqref{eq:GDoF-region} have an intuitive interpretation.
Let us denote by $\tilde{\mv}=(m_{1}, m_{2}, \dots, m_{K})$ a channel state where Receiver $k$ falls into State $m_{k} \in [M]$. In this way, we have constructed in total $M^K$ states (in addition to the $M$ original states defined in the system model, we further have $M^K - K$ \emph{mixed} states where the receivers belong to different original states). As the capacity region only depends on marginals, these additional mixed states do not hurt the capacity (a detailed argument appears in Lemma \ref{lemma:aux-state-decode}).
Now \eqref{eq:tin-optimal} says that the TIN optimality condition for regular interference channel \cite{Geng_TIN} should hold for \emph{every single one of the $M^K$ states} and \eqref{eq:GDoF-region} is the collection of inequalities that constitute the GDoF region for \emph{each individual state}. A concrete illustration appears in Example \ref{ex:main-ex}.
\end{remark}

\begin{remark} \normalfont
The compound setting studied in \cite{TIN-Compound} is a special case of ours.
By letting $\pi_k(m)=1$ for all $k, m$, all receivers are supposed to decode only the basic messages over all states, and our system model reduces to the compound setting in \cite{TIN-Compound}. Setting $d_k^{[m]}=0$ for all $m \ge 2$, the GDoF region in \eqref{eq:GDoF-region} recovers that in \cite{TIN-Compound}.
\end{remark}

\begin{remark}\normalfont
\label{remark:power-control}
A natural choice of the number of messages to decode at a given state, $\pi_k(m')$ is
$\pi_k(m') := \left|\Big\{m \in [M]:  \alpha_{kk}^{[m]} -\max_{j: j \ne k} \{ \alpha_{kj}^{[m]}\} < \alpha_{kk}^{[m']} -\max_{j: j \ne k} \{ \alpha_{kj}^{[m']}\}\Big\}\right|+1,$
where we use $\big\{\alpha_{kk}^{[m]} -\max_{j: j \ne k} \{ \alpha_{kj}^{[m]}\}\big\}$ to reflect the TIN decoding capability for Receiver $k$ at State $m$. 

Subject to this choice of $\pi_k(m')$,
the transmit power exponents $\{r_k^{[m]}\}_{k,m}$ can be computed as follows. 
Let 
$m_k := \arg \max_{m \in [M]} \Big\{ \alpha_{kk}^{[m]} -\max_{j: j \ne k} \{ \alpha_{kj}^{[m]} \} \Big\}, %
\forall k \in [K],$
and consider an auxiliary interference network where Receiver $k, k\in [K]$ is statistically equivalent to that at State $m_k$.
Next, given a feasible GDoF tuple $(\{d_k^{[m]}\}_{k,m})$, the optimal power allocation exponents of the basic messages $\{r_k^{[1]}\}_{k}$ can be obtained by applying the power control algorithms in \cite{TIN-Compound, Yi-TIN} to the auxiliary interference network 
with the GDoF tuple 
$(\sum_{m=1}^{M}d_1^{[m]}, \sum_{m=1}^{M}d_2^{[m]}, \dots, \sum_{m=1}^{M}d_K^{[m]}).$
The power exponents of the opportunistic messages can then be successively computed according to
$r_k^{[m+1]} = r_k^{[m]} - d_k^{[m]},
m=1, \dots, M-1.$
\end{remark}

In what follows, we consider a typical example to illustrate our result and remarks.

\begin{example} \normalfont
\label{ex:main-ex}
We hereby consider a 3-user interference channel with 2 states as shown in Fig. \ref{fig:ex3}.  For the sake of notational clarity, we denote by $S_1$ and $S_2$ two states respectively, by $W_k$ and $d_k=d_k^{[1]}$ the basic message and its GDoF, respectively, and by $\Delta W_k$ and $\Delta d_k=d_k^{[2]}$ the opportunistic message and its GDoF, respectively. The transmitted signal is produced by using simple layered superposition coding of the messages $W_k$ and $\Delta W_k$ with respective power exponent $r_k$ and $\Delta r_k$.
\begin{figure}[h]
\center
\includegraphics[width=4.5 in]{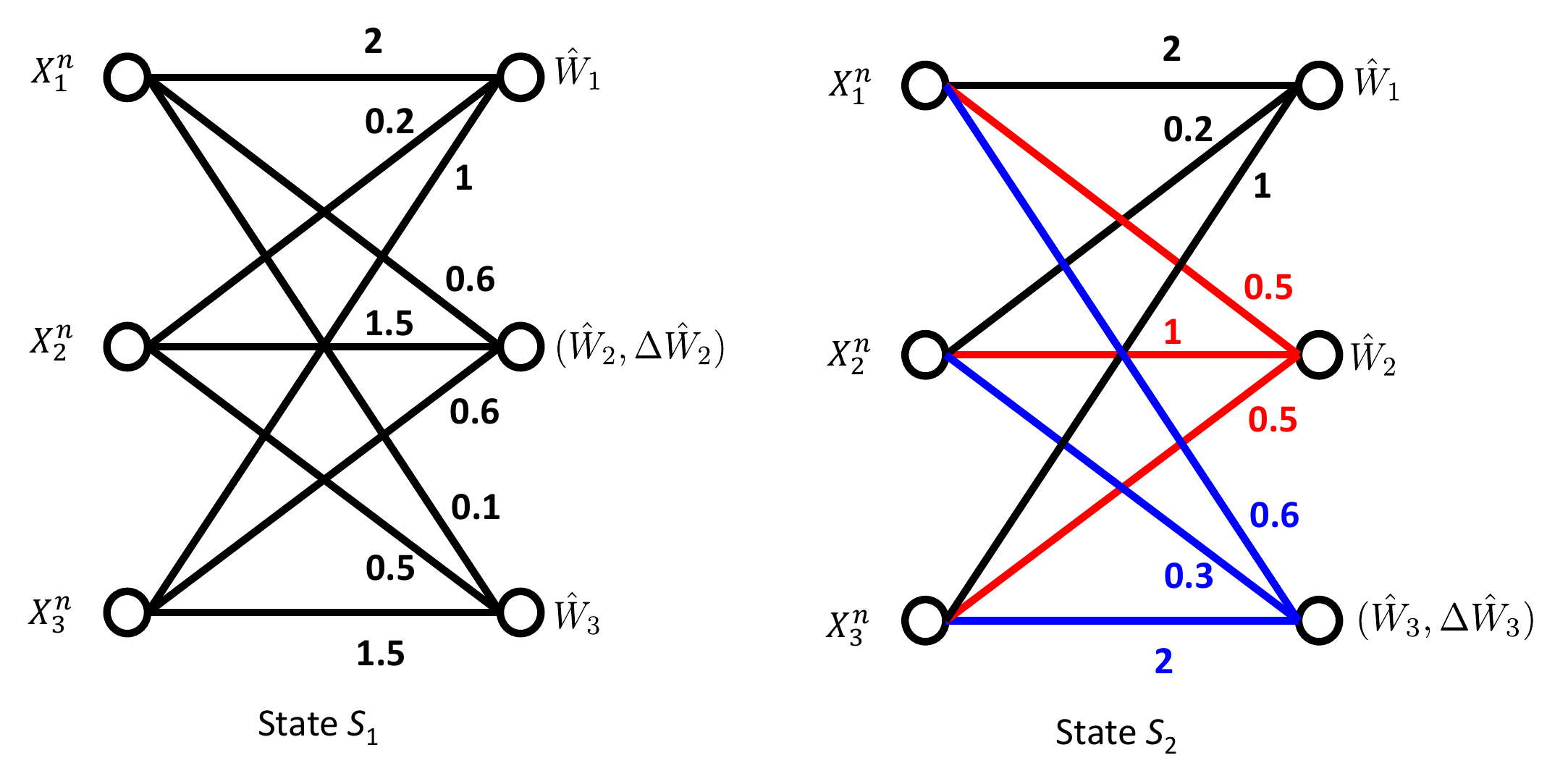}
\caption{\small A 3-user interference network with 2 states.}
\label{fig:ex3}
\end{figure}

We use the choice of $\pi_k(m')$ as stated in Remark \ref{remark:power-control} (Theorem \ref{thm:main} holds for any choice. We pick a specific choice here to illustrate the result). That is, we compute the difference of the desired signal strength and the strong interference strength level for each receiver at each state, as follows.
\begin{subequations}
\begin{align}
&\text{Receiver 1: } \quad \; \; \quad  2-\max\{0.2,1\} = 2 - \max\{0.2, 1\}, \\
&\text{Receiver 2: } \quad 1.5-\max\{0.6,0.6\} \ge 1 - \max\{0.5, 0.5\}, \\ &\text{Receiver 3: } \quad 1.5 - \max\{0.1,0.5\} \le 2 - \max\{0.6,0.3\}.
\end{align}
\end{subequations}
According to the relative strength of the two states for each receiver, we set
\begin{align}
\pi_1 (1)=1, \quad \pi_2 (1)=2, \quad \pi_3 (1)=1, \\
\pi_1 (2)=1, \quad \pi_2 (2)=1, \quad \pi_3 (2)=2,
\end{align}
where for example, $\pi_2 (1)=2$ because for Receiver 2, the first state is in a better condition such that we wish to decode both the basic and the opportunistic messages.

Next, we check the TIN-optimality condition and find the GDoF region. To this end, following Remark \ref{remark:remark1}, we construct two auxiliary states $S'_{:,1,2}$ and $S'_{:,2,1}$. Note that the channels to Receiver 1 remain the same across the two states.
As such, the original and auxiliary states with respect to channel strength exponents are given by
\begin{subequations}
\begin{align}
S_1 &:= \{2, 0.2, 1; 0.6,1.5,0.6; 0.1,0.5,1.5\}\\
S'_{:,1,2} &:= \{2, 0.2, 1; 0.6,1.5,0.6; 0.6,0.3,2\}\\
S'_{:,2,1} &:= \{2, 0.2, 1; 0.5,1,0.5; 0.1,0.5,1.5\}\\
S_2 &:= \{2, 0.2, 1; 0.5,1,0.5; 0.6,0.3,2\}.
\end{align}
\end{subequations}

It can be verified that the TIN-optimality condition is satisfied for every original and auxiliary channel state, so the TIN optimality condition \eqref{eq:tin-optimal} holds for our setting. Thus, Theorem \ref{thm:main} applies. According to \eqref{eq:GDoF-region}, after removing the redundant inequalities, we have the optimal GDoF region
\begin{align*}
d_1 \ge 0, \; d_2 \ge 0, \; d_3 &\ge 0\\
\Delta d_1 \ge 0, \; \Delta d_2 \ge 0, \; \Delta d_3 &\ge 0\\
d_1 &\le 2 \\
d_2 &\le 1 \\
d_2 + \Delta d_2  &\le 1.5 \\
d_3 &\le 1.5\\
d_3 + \Delta d_3 &\le 2 \\
d_1 + d_2 &\le 2.3 \\
d_1 + d_2 + \Delta d_2 &\le 2.7 \\
d_1 + d_3 + \Delta d_3 &\le 2.4 \\
d_2 + d_3 &\le 1.5 \\
d_2  + d_3 + \Delta d_2 &\le 1.9 \\
d_2 + d_3 + \Delta d_3 &\le 2.2 \\
d_2 + d_3 + \Delta d_2 + \Delta d_3 &\le 2.6 \\
d_1+d_2+d_3 &\le 2.5\\
d_1 + d_2 + d_3 + \Delta d_2 &\le 2.9 \\
d_1 + d_2 + d_3 + \Delta d_3 &\le 3.2\\
d_1 + d_2 + d_3 + \Delta d_2 + \Delta d_3 &\le 3.6.
\end{align*}

As stated in Remark \ref{remark:remark1}, the above optimal GDoF region can also be obtained by collecting the individual and sum GDoF inequalities from the GDoF region of all original and auxiliary states. For each state, we have 
\begin{align}
&\Pmatrix{S_1:  \left\{\Pmatrix{d_1 \le 2 \\ d_2+\Delta d_2 \le 1.5 \\ d_3 \le 1.5 \\d_1+d_2+\Delta d_2 \le 2.7 \\ d_1+d_3 \le 2.4 \\ d_2+\Delta d_2+d_3 \le 1.9 \\ d_1+d_2+\Delta d_2+d_3 \le 2.9}\right.}
\Pmatrix{S_2:  \left\{\Pmatrix{d_1 \le 2 \\ d_2 \le 1 \\ d_3+\Delta d_3 \le 2 \\d_1+d_2 \le 2.3 \\ d_1+d_3+\Delta d_3 \le 2.4 \\ d_2+d_3+\Delta d_3 \le 2.2 \\ d_1+d_2+d_3+\Delta d_3 \le 3.2}\right.}\\
&\Pmatrix{S'_{:,1,2}:  \left\{\Pmatrix{d_1 \le 2 \\ d_2+\Delta d_2 \le 1.5 \\ d_3+\Delta d_3 \le 2 \\d_1+d_2+\Delta d_2 \le 2.7 \\ d_1+d_3+\Delta d_3 \le 2.4 \\ d_2+\Delta d_2+d_3+\Delta d_3 \le 2.6 \\ d_1+d_2+\Delta d_2+d_3+\Delta d_3 \le 3.6}\right.}
\Pmatrix{S'_{:,2,1}:  \left\{\Pmatrix{d_1 \le 2 \\ d_2 \le 1 \\ d_3 \le 1.5 \\d_1+d_2 \le 2.3 \\ d_1+d_3 \le 2.4 \\ d_2+d_3 \le 1.5 \\ d_1+d_2+d_3 \le 2.5}\right.}.
\end{align}
It is easy to check that the collection of all above inequalities gives us the final optimal GDoF region.

{The achievability of the above GDoF region can be verified by checking the existence  of power exponents $r_k$'s for all extreme points. 
For instance, the GDoF tuple $(d_1, d_2, d_3, \Delta d_2, \Delta d_3) = (2,0.3,0.2,0.4,0.2)$ is one of the nontrivial extreme points.
Following Remark \ref{remark:power-control}, we consider the auxiliary state $S'_{:,1,2}$ which has the maximum TIN decoding capability at each receiver.  Applying the power control algorithms in \cite{TIN-Compound, Yi-TIN} to State $S'_{:,1,2}$ with the GDoF tuple $(d_1, d_2+\Delta d_2, d_3+\Delta d_3)=(2,0.7,0.4)$, we obtain $(r_1,r_2,r_3)=(0, -0.2, -1)$ and $(\Delta r_2, \Delta r_3)=(r_2,r_3)-(d_2,d_3)=(-0.5,-1.2)$. It is not hard to verify that all messages are successfully decoded with such power allocation.}
\end{example}

In what follows, we present the proofs of the achievability and the converse.

\section{Achievability}
For the achievability, to illustrate the main idea, we first take a 3-user interference channel with 2 states as an example (see Fig. \ref{fig:3-user-potential}(a)) in Section \ref{sec:3-user}. %

\subsection{A 3-user Example}\label{sec:3-user}
To simplify the notation, we denote by $W_k$ and $\Delta W_k$ the basic and opportunistic messages, respectively. Given a state $m \in \{1,2\}$, if $\pi_k(m)=1$, then Receiver $k$ only needs to decode the basic message $W_k$, and otherwise if $\pi_k(m)=2$, both basic and opportunistic messages are required to be decoded.
We use simple layered superposition coding at the transmitters and opportunistic TIN at the receivers (as introduced in Section \ref{sec:gdof-framework}) to derive the achievable GDoF region.

At Transmitter $k$, we send the superposition of the Gaussian coded basic message $W_k$ and the Gaussian coded opportunistic message $\Delta W_k$ with transmit power exponents $r_k^{[1]}$ and $r_k^{[2]}$ respectively.
\begin{align}
X_k = P^{r_k^{[1]}}X_b(W_k) + P^{r_k^{[2]}} X_o(\Delta W_k)
\end{align}
where 
\begin{align} \label{eq:power-difference}
r_k^{[m]} - r_k^{[m+1]}=d_k^{[m]}, \quad \forall k \in [3], \; \forall m \in [2]
\end{align}
{and $r_k^{[3]} = r_k^{[1]} - d_1^{[1]} - d_2^{[2]}$ is the lowest power level used by Transmitter $k$.}

By ignoring the $\max\{0,\cdot\}$ term in (\ref{eq:GDoF}), we focus on the achievable GDoF via polyhedral TIN (as done in \cite{Geng_TIN, TIN-X, Sun_Jafar_ParallelTIN, TIN-Compound, TIN-IMAC}), for which
\begin{align} \label{eq:gdof-constraint}
d_k^{[m]} = \min \Big\{r_k^{[m]} - r_k^{[m+1]}, \min_{m': \pi_k(m') \ge m} \big\{\alpha_{kk}^{[m']} + r_k^{[m]} - \max \{0, \max_{j: j \neq k} \alpha_{kj}^{[m']} + r_j^{[1]}\} \big\} \Big\} \geq 0.
\end{align}
Thus, the achievable GDoF region $\Pc$ by polyhedral TIN is the set of GDoF tuples $(d_k^{[m]},k\in[3], m \in [2])$ for which there exist $\{r_k^{[m]},k\in[3], m \in [2]\}$ such that the above constraints \eqref{eq:gdof-constraint} are satisfied for all $k \in [3]$ and $m \in [2]$.

Denote the GDoF region in (\ref{eq:GDoF-region}) by $\mathcal{P}^*$. To show that $\mathcal{P}^*$ is achievable by polyhedral TIN, we construct an achievable GDoF region $\Pc'$ such that $\Pc' \subseteq \Pc$ and $\Pc' = \mathcal{P}^*$.

\subsubsection{Constructing  $\Pc' \subseteq \Pc$}
By imposing \eqref{eq:power-difference} in \eqref{eq:gdof-constraint}, we have an achievable GDoF region $\Pc'$ that is a subset of $\Pc$.
Plugging \eqref{eq:power-difference} into \eqref{eq:gdof-constraint}, we have
\begin{align}
0 \le d_k^{[m]} &\le \min_{m': \pi_k(m') \ge m} \Big\{\alpha_{kk}^{[m']} + r_k^{[m]} - \max \big\{0, \max_{j: j \neq k} \alpha_{kj}^{[m']} + r_j^{[1]}\big\} \Big\}, \quad \forall m'\in\{1,2\}.
\end{align}
Specifying all possible values of $m, m'$, we have
\begin{subequations}
\begin{align}
d_k^{[1]} &\le \alpha_{kk}^{[1]} + r_k^{[1]} - \max \big\{0, \max_{j: j \neq k} \alpha_{kj}^{[1]} + r_j^{[1]}\big\} \label{eq:rr1}\\
d_k^{[1]} &\le \alpha_{kk}^{[2]} + r_k^{[1]} - \max \big\{0, \max_{j: j \neq k} \alpha_{kj}^{[2]} + r_j^{[1]}\big\} \label{eq:rr2}\\
\text{if } \pi_k(1)=2, \quad d_k^{[2]} &\le \alpha_{kk}^{[1]} + r_k^{[2]} - \max \big\{0, \max_{j: j \neq k} \alpha_{kj}^{[1]} + r_j^{[1]}\big\} \\
\overset{\eqref{eq:power-difference}}{\Longleftrightarrow} & \quad d_k^{[1]}+ d_k^{[2]} \le \alpha_{kk}^{[1]} + r_k^{[1]} - \max \big\{0, \max_{j: j \neq k} \alpha_{kj}^{[1]} + r_j^{[1]} \big\} \label{eq:3-user-2-state-relax1}\\
\text{if } \pi_k(2)=2, \quad d_k^{[2]} &\le \alpha_{kk}^{[2]} +  r_k^{[2]} - \max \big\{0, \max_{j: j \neq k} \alpha_{kj}^{[2]} + r_j^{[1]}\big\} \\
\overset{\eqref{eq:power-difference}}{\Longleftrightarrow} & \quad d_k^{[1]}+  d_k^{[2]} \le \alpha_{kk}^{[2]} + r_k^{[1]} - \max \big\{0, \max_{j: j \neq k} \alpha_{kj}^{[2]} + r_j^{[1]} \big\}. \label{eq:3-user-2-state-relax2}
\end{align}
\end{subequations}
Note that \eqref{eq:3-user-2-state-relax1} implies (\ref{eq:rr1}) and \eqref{eq:3-user-2-state-relax2} implies (\ref{eq:rr2}). %
Combining all above inequalities, we have the compact form
\begin{align}
\sum_{m=1}^{\pi_k(m')}d_k^{[m]} &\le \alpha_{kk}^{[m']} + r_k^{[1]} - \max\{0, \max_{j: j\neq k}\{\alpha_{kj}^{[m']} + r_j^{[1]}\}\}, \quad \forall m' \in \{1,2\}.
\end{align}
That is, we have a GDoF region $\Pc'$ of the set of GDoF tuples $(\{d_k^{[m]}\}_{k,m})$ with respect to $r_k^{[m]}$'s
\begin{subequations}
\begin{align} 
r_k &\le 0, \quad \forall k \in [3]\\
d_k^{[m]} &\ge 0, \quad \forall k \in [3], m \in [2]\\
d_k^{[m]} &= r_k^{[m]} - r_k^{[m+1]}, \quad \forall k \in [3], \; m\in [2] \label{eq:power-difference-m=1}\\
\sum_{m=1}^{\pi_k(1)}d_k^{[m]} &\le \alpha_{kk}^{[1]} + r_k^{[1]} - \max\{0, \max_{j: j\neq k}\{\alpha_{kj}^{[1]} + r_j^{[1]}\}\}, \quad k \in [3] \label{eq:m'=1m=1}\\
\sum_{m=1}^{\pi_k(2)}d_k^{[m]} &\le \alpha_{kk}^{[2]} + r_k^{[1]} - \max\{0, \max_{j: j\neq k}\{\alpha_{kj}^{[2]} + r_j^{[1]}\}\}, \quad k \in [3] \label{eq:m'=2m=2}
\end{align}
\end{subequations}
which is not larger than $\Pc$ because of the additional constraint \eqref{eq:power-difference-m=1}, i.e., $\Pc' \subseteq \Pc$. For notational simplicity, we hereafter set $r_k^{[1]}=r_k$ and $r_k^{[2]}=\Delta r_k$.

\subsubsection{Proof of $\Pc' = \mathcal{P}^*$}

To show $\Pc' = \mathcal{P}^*$, we %
eliminate the $r_k$'s in $\Pc'$, following the idea in \cite{Geng_TIN}.\footnote{Similar ideas have been applied and extended to other scenarios \cite{Geng_TIN, TIN-Compound, TIN-X, TIN-IMAC} to tackle different message settings and network topologies.}
Specifically, we construct a potential digraph where the lengths of the arcs are represented only by $d_k$'s and $\alpha_{ij}$'s. Then we verify the existence of a potential function by imposing that the lengths of directed circuits in the potential digraph are non-negative. 

\subsubsection*{\underline{Potential Digraph}}
For a given GDoF tuple $(d_k^{[m]},k\in[3], m \in [2]) \in \RR_+^6$ in $\Pc'$, according to \eqref{eq:m'=1m=1} and \eqref{eq:m'=2m=2}, it is feasible if and only if there exist $r_k$'s for all $k \in [3]$ satisfying, 
 \begin{subequations} \label{eq:ex-r-inequalities}
\begin{align}
r_k &\le 0\\
r_k &\ge \sum_{m=1}^{\pi_k(1)}d_k^{[m]}  - \alpha_{kk}^{[1]}\\
r_k &\ge \sum_{m=1}^{\pi_k(2)}d_k^{[m]}  - \alpha_{kk}^{[2]}\\
r_k -r_j &\ge \Big(\sum_{m=1}^{\pi_k(1)}d_k^{[m]} - \alpha_{kk}^{[1]}\Big) + \alpha_{kj}^{[1]},  \quad \forall j \neq k \\
r_k -r_j &\ge \Big(\sum_{m=1}^{\pi_k(2)}d_k^{[m]}  - \alpha_{kk}^{[2]}\Big) + \alpha_{kj}^{[2]},  \quad \forall j \neq k.
\end{align} 
\end{subequations}

In view of these inequalities, we construct a simple potential digraph $D'=(V',A')$, where
\begin{align}
V' &= \{u, v_1^{[1]}, v_1^{[2]}, v_2^{[1]}, v_2^{[2]}, v_3^{[1]},v_3^{[2]}\} \\
A' &=\{(u,v,l): u,v \in V', l \in \RR\}.
\end{align} 
The arc set consists of four parts $A' = A'_1 \cup A'_2 \cup A'_3 \cup A'_4$, where
\begin{subequations}
\begin{align}
A'_1 &= \{(u,v_k^{[m]},l(u,v_k^{[m]})): k \in [3], m \in [2]\} \\
A'_2 &= \{(v_k^{[m]},u,l(v_k^{[m]},u)): m \in \{1,2\}, k \in [3]\}\\
A'_3 &= \{(v_k^{[m]},v_j^{[m]},l(v_k^{[m]},v_j^{[m]})): m \in \{1,2\}, k, j \in [3], k \ne j\} \\
A'_4 &= \{(v_k^{[m_1]},v_k^{[m_2]},l(v_k^{[m_1]},v_k^{[m_2]})): m_1,m_2 \in \{1,2\}, m_1 \ne m_2, k \in [3]\}
\end{align}
\end{subequations}
with a length $l(a,b)$ assigned to every single arc $(a,b) \in \Ac'$ as follows:
\begin{subequations}
\begin{align}
l(u,v_k^{[m]}) &= 0, \quad \forall m=\{1,2\}\\
l(v_k^{[1]},u) &= \alpha_{kk}^{[1]} - \sum_{m=1}^{\pi_k(1)}d_k^{[m]}\\
l(v_k^{[2]},u) &= \alpha_{kk}^{[2]} - \sum_{m=1}^{\pi_k(2)}d_k^{[m]} \\
l(v_k^{[1]},v_j^{[1]}) &= \alpha_{kk}^{[1]} - \sum_{m=1}^{\pi_k(1)}d_k^{[m]} - \alpha_{kj}^{[1]}, \quad \forall k \ne j\\
l(v_k^{[2]},v_j^{[2]}) &= \alpha_{kk}^{[2]} - \sum_{m=1}^{\pi_k(2)}d_k^{[m]} - \alpha_{kj}^{[2]}, \quad \forall k \ne j\\
l(v_k^{[m_1]},v_k^{[m_2]}) &= 0, \quad \forall m_1,m_2 \in \{1,2\}, m_1 \ne m_2.
\end{align}
\end{subequations}

An illustrative example on the simple potential digraph when $\pi_k(1) = 1$ and $\pi_k(2) = 2$ is shown in Fig. \ref{fig:3-user-potential}.

\begin{figure}[h]
\center
\includegraphics[width=5 in]{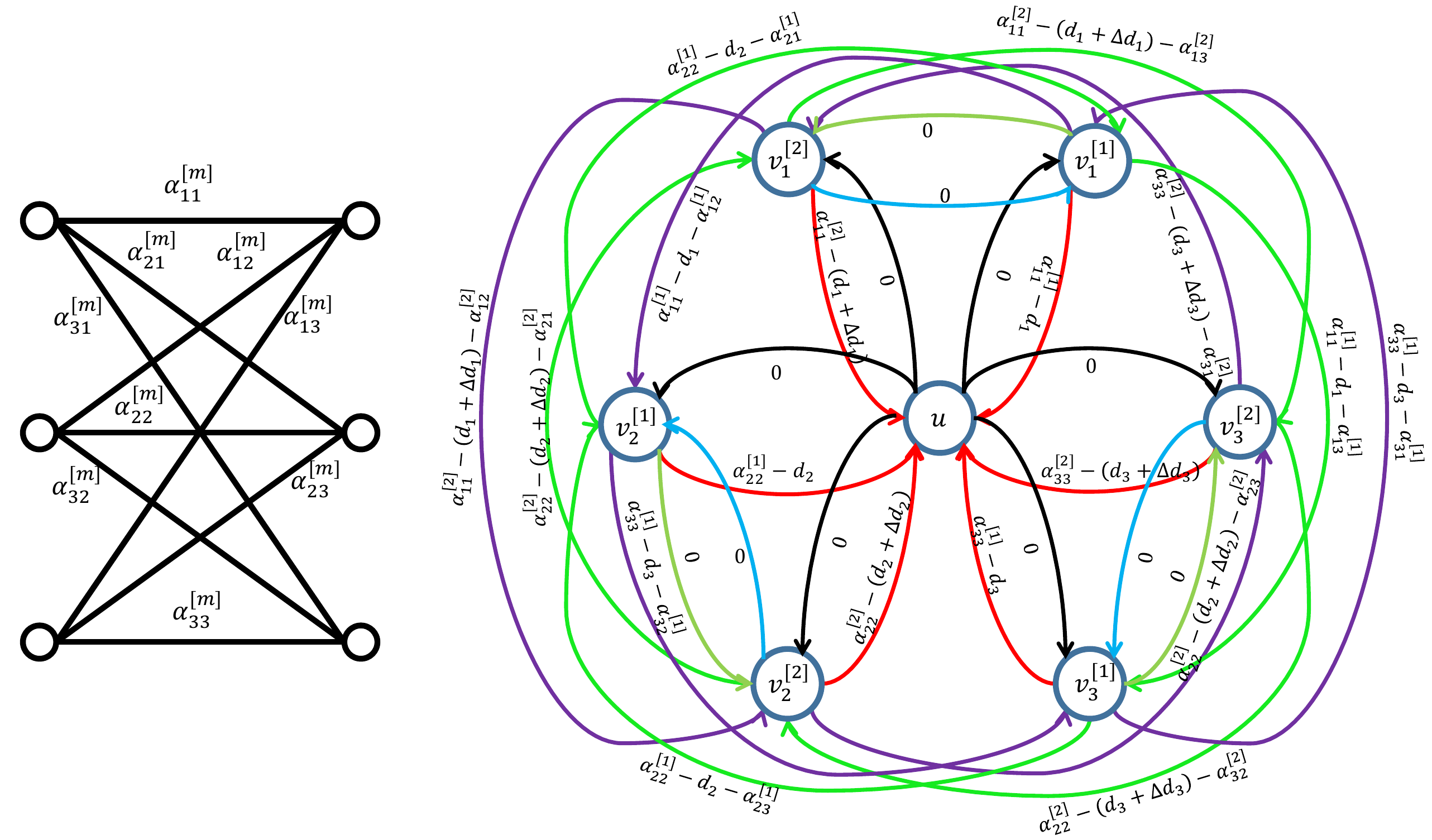}
\caption{\small A 3-user interference network (left) with 2 states $m \in \{1,2\}$ where channel strength exponents vary across two states, and a simple potential digraph (right) corresponding to a special case when $\pi_k(1)=1$ and $\pi_k(2)=2$ for all $k=\{1,2,3\}$. Receiver $k$ decodes the basic message $W_k$ yielding GDoF $d_k$ at State $S_1$, and decodes at State $S_2$ both basic and opportunistic messages $W_k$ and $\Delta W_k$ yielding GDoF $d_k$ and $\Delta d_k$ respectively.}
\label{fig:3-user-potential}
\end{figure}

By the potential digraph, we connect the existence of $r_k$'s to the existence of a valid potential function for this digraph.

\begin{lemma} \label{lemma:potential-func}
The GDoF tuple $(\{d_k^{[k,m]}\}_{k,m}) \in \RR_+^6$ is feasible if and only if there exists a valid potential function for the simple digraph $D'=(V',A')$.
\end{lemma}
\begin{proof}
The proof is similar to that in \cite{Geng_TIN}. Given a simple digraph $D=(V,A)$, a function $p: V \mapsto \RR$ is called potential if for every arc $(a,b) \in A$ with length $l(a,b)$, it satisfies $l(a,b) \ge p(b)-p(a)$.

In the simple digraph $D'=(V',A')$, if there exists a valid potential function $p(\cdot)$, then letting 
\begin{align}
p(u)=0, \quad p(v_k^{[1]}) = p(v_k^{[2]}) = r_k, \forall k
\end{align}
the potential function values must satisfy the following inequalities
\begin{subequations}
\begin{align}
l(v_k^{[1]}, v_j^{[1]}) \ge r_j - r_k &\Longleftrightarrow r_k - r_j \ge \sum_{m=1}^{\pi_k(1)}d_k^{[m]}-\alpha_{kk}^{[1]}+\alpha_{kj}^{[1]} \\
l(v_k^{[2]}, v_j^{[2]}) \ge r_j - r_k &\Longleftrightarrow r_k - r_j \ge \sum_{m=1}^{\pi_k(2)}d_k^{[m]}-\alpha_{kk}^{[2]}+\alpha_{kj}^{[2]} \\
l(u, v_k^{[1]}) \ge r_k &\Longleftrightarrow r_k \le 0 \\
l(u, v_k^{[2]}) \ge r_k &\Longleftrightarrow r_k \le 0 \\
l(v_k^{[1]},u) \ge  -r_k &\Longleftrightarrow r_k \ge \sum_{m=1}^{\pi_k(1)}d_k^{[m]} - \alpha_{kk}^{[1]}\\
l(v_k^{[2]},u) \ge  -r_k &\Longleftrightarrow r_k \ge \sum_{m=1}^{\pi_k(2)}d_k^{[m]}k - \alpha_{kk}^{[2]}\\
l(v_k^{[1]},v_k^{[2]}) \ge r_k-r_k &\Longleftrightarrow 0 \ge 0 \\
l(v_k^{[2]},v_k^{[1]}) \ge r_k-r_k &\Longleftrightarrow 0 \ge 0.
\end{align}
\end{subequations}
It can be readily verified that the nontrivial inequalities above exactly match those in \eqref{eq:ex-r-inequalities}. Both``if'' and ``only if'' parts hold together.
\end{proof}

The above simple potential digraph consists of $MK+1=7$ vertices for a 2-state 3-user interference channel, which becomes involved for large $M, K$. Next, we simplify it to a labeled multi-digraph.
\subsubsection*{\underline{Labeled Multi-digraph Representation} }
We construct a labeled multi-digraph $D=(V,A)$ to represent the simple digraph towards simplifying it. Fig. \ref{fig:multi-digraph} gives an illustrative example on the simplification of the potential digraph in Fig. \ref{fig:3-user-potential}.
\begin{figure}[h]
\center
\includegraphics[width=3 in]{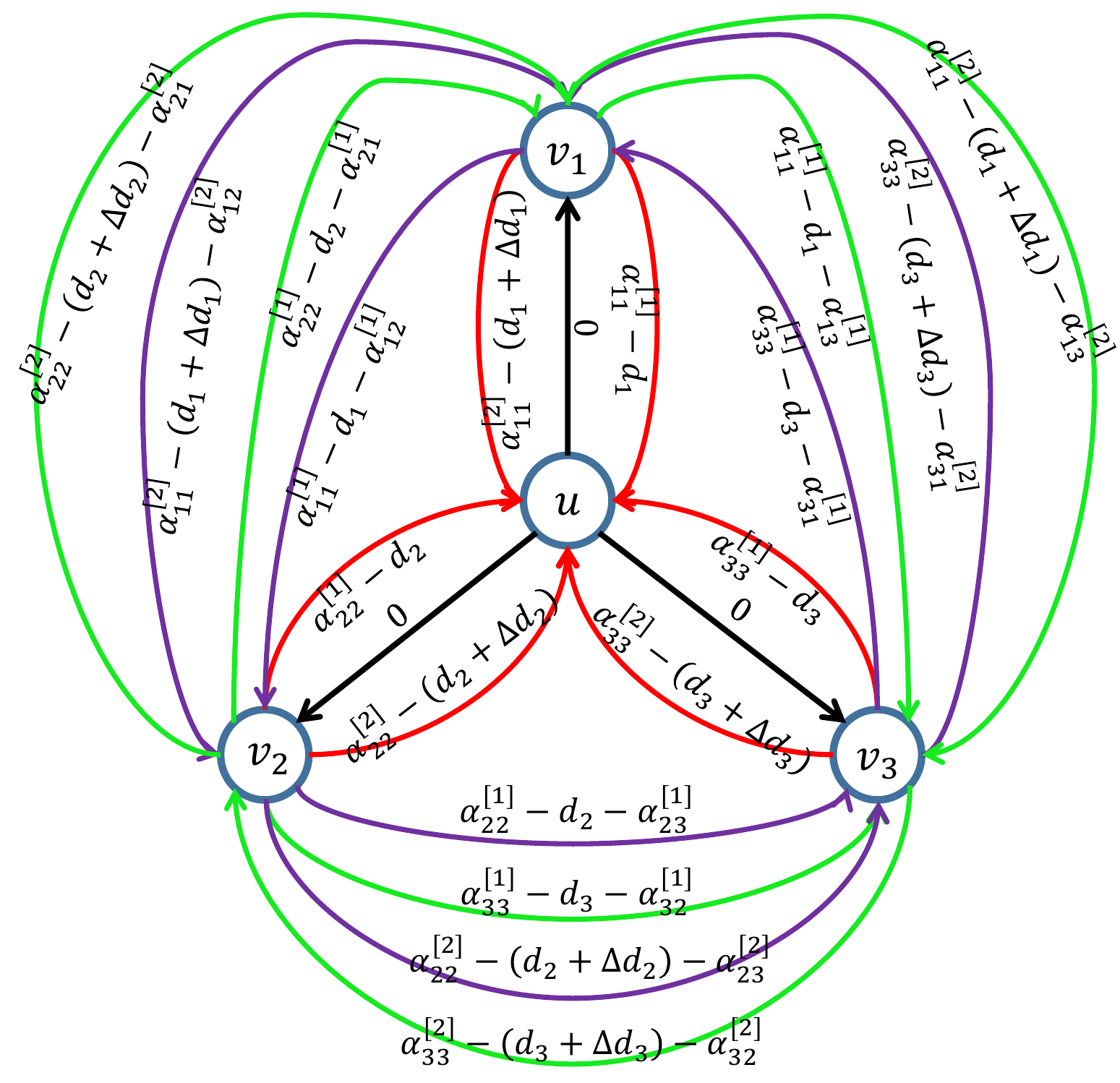}
\caption{\small The labeled multi-digraph for a 3-user interference network with 2 states simplifies the simple digraph in Fig. \ref{fig:3-user-potential}.}
\label{fig:multi-digraph}
\end{figure}

Given the simple potential graph $D'=(V',A')$, we merge the vertices $v_k^{[1]}$ and $v_k^{[2]}$ into a single one $v_k$, and the arcs in $A'$ are labeled as follows: (1) the arcs $\{(u,v_k^{[m]},l(u,v_k^{[m]})): m=\{1,2\}\}$ are merged as a single arc $(u,v_k,l(u,v_k))$; (2) the arcs between $\{v_k^{[m]}: m=\{1,2\}\}$ are removed; (3) the arcs from $v_k^{[m]}$ to $u$ are relabeled as $(v_k,u,m,l^{[m]}(v_k,u))$; (4) the arcs $(v_k^{[m]},v_j^{[m]},l(v_k^{[m]},v_j^{[m]}))$ are relabeled as $(v_k,v_j,m,l^{[m]}(v_k,v_j))$.

In particular, the labeled multi-digraph $D=(V,A)$ is such that
\begin{align}
V &= \{u, v_1, v_2, v_3\} \\
A &=\{(u,v,m,l): u,v \in V, m \in \{1,2\}, l \in \RR\}
\end{align}
where $m$ specifies the label of an arc and $l=l^{[m]}(u,v)$ is the length between $u,v \in V$ with label $m$. The arc set consists of three parts $A = A_1 \cup A_2 \cup A_3$, where
\begin{subequations}
\begin{align}
A_1 &= \{(u,v_k,,l(u,v_k)): k \in [3]\} \\
A_2 &= \{(v_k,u,m,l^{[m]}(v_k,u)): m \in \{1,2\}, k \in [3]\}\\
A_3 &= \{(v_k,v_j,m,l^{[m]}(v_k,v_j)): m \in \{1,2\}, k, j \in [3], k \ne j\} 
\end{align}
\end{subequations}
with a length $l^{[m]}(a,b)$ assigned to every single arc $(a,b) \in \Ac$ as follows:
\begin{subequations}
\begin{align}
l(u,v_k) &= 0\\
l^{[1]}(v_k,u) &= \alpha_{kk}^{[1]} - \sum_{m=1}^{\pi_k(1)}d_k^{[m]}\\
l^{[2]}(v_k,u) &= \alpha_{kk}^{[2]} - \sum_{m=1}^{\pi_k(2)}d_k^{[m]} \\
l^{[1]}(v_k,v_j) &= \alpha_{kk}^{[1]} - \sum_{m=1}^{\pi_k(1)}d_k^{[m]}- \alpha_{kj}^{[1]} \\
l^{[2]}(v_k,v_j) &= \alpha_{kk}^{[2]} - \sum_{m=1}^{\pi_k(2)}d_k^{[m]} - \alpha_{kj}^{[2]}.
\end{align}
\end{subequations}

This multi-digraph representation simplifies the description of the potential graph, due to the following lemma.
\begin{lemma} \label{lemma:digraph-equiv}
The labeled multi-digraph $D$ inherits two properties from the single digraph $D'$: (1) the potential function in $D'$ is valid in $D$; (2) for every circuit in $D'$ there is a corresponding circuit in $D$ with the same length.
\end{lemma}
\begin{proof}
In the labeled multi-digraph $D=(V,A)$, using the same potential function $p(\cdot)$ as in $D'$, we assign the following values
\begin{align}
p(u) = 0, \quad p(v_k)=r_k, \forall k.
\end{align}
As the potential function values satisfy the same set of inequalities, 
\begin{subequations}
\begin{align}
l^{[1]}(v_k, v_j) \ge r_j - r_k &\Longleftrightarrow r_k - r_j \ge \sum_{m=1}^{\pi_k(1)}d_k^{[m]}-\alpha_{kk}^{[1]}+\alpha_{kj}^{[1]} \\
l^{[2]}(v_k, v_j) \ge r_j - r_k &\Longleftrightarrow r_k - r_j \ge \sum_{m=1}^{\pi_k(2)}d_k^{[m]} - \alpha_{kk}^{[2]}+\alpha_{kj}^{[2]} \\
l(u, v_k) \ge r_k &\Longleftrightarrow r_k \le 0 \\
l^{[1]}(v_k,u) \ge  -r_k &\Longleftrightarrow r_k \ge \sum_{m=1}^{\pi_k(1)}d_k^{[m]} - \alpha_{kk}^{[1]}\\
l^{[2]}(v_k,u) \ge  -r_k &\Longleftrightarrow r_k \ge \sum_{m=1}^{\pi_k(2)}d_k^{[m]} - \alpha_{kk}^{[2]},
\end{align}
\end{subequations}
we conclude that the potential function in $D'$ works in $D$.

For the correspondence of directed circuits in $D'$ and $D$, we illustrate some of the typical ones in Table \ref{table:dicircuits}. Note that the list therein is not exhaustive.
\begin{table}[h]
  \centering
  \caption{The correspondence of directed circuits between $D$ and $D'$.}
  \begin{tabular}{c | c | c }
    \hline \hline
      $D'$ & $D$ & Length  \\ \hline \hline
     $(u,v_{1}^{[1]},u)$ or $(u,v_{1}^{[1]},v_{1}^{[2]},u)$ & $(u\rightarrow v_1 \xrightarrow{[1]} u)$ & $l^{[1]}(v_1,u)$   \\ \hline
     $(u,v_{1}^{[1]},v_{2}^{[1]},u)$ & $(u\rightarrow v_1 \xrightarrow{[1]} v_2 \xrightarrow{[1]} u)$ &   $l^{[1]}(v_1,v_2)+l^{[1]}(v_2,u)$ \\ \hline
      $(u,v_{1}^{[1]},v_{1}^{[2]},v_{2}^{[2]},u)$ & $(u\rightarrow v_1 \xrightarrow{[2]} v_2 \xrightarrow{[2]} u)$ &   $l^{[2]}(v_1,v_2)+l^{[2]}(v_2,u)$ \\ \hline
      $(u,v_{1}^{[1]},v_{2}^{[1]},v_{2}^{[2]},u)$ & $(u\rightarrow v_1 \xrightarrow{[1]} v_2 \xrightarrow{[2]} u)$ &   $l^{[1]}(v_1,v_2)+l^{[2]}(v_2,u)$ \\ \hline
      $(v_{1}^{[1]},v_{2}^{[1]},v_{1}^{[1]})$ & $(v_1 \xrightarrow{[1]} v_2 \xrightarrow{[1]} v_1)$ &   $l^{[1]}(v_1,v_2)+l^{[1]}(v_2,v_1)$ \\ \hline
      $(v_{1}^{[1]},v_{2}^{[1]},v_{2}^{[2]},v_{2}^{[1]},v_{1}^{[1]})$ & $(v_1 \xrightarrow{[1]} v_2 \xrightarrow{[2]} v_1)$ &   $l^{[1]}(v_1,v_2)+l^{[2]}(v_2,v_1)$ \\ \hline
      $(v_{1}^{[1]},v_{2}^{[1]},v_3^{[1]},v_{1}^{[1]})$ & $(v_1 \xrightarrow{[1]} v_2 \xrightarrow{[1]} v_3 \xrightarrow{[1]} v_1)$ &   $l^{[1]}(v_1,v_2)+l^{[1]}(v_2,v_3)+l^{[1]}(v_3,v_1)$ \\ \hline
      $(v_{1}^{[1]},v_{2}^{[1]}, v_{2}^{[2]},v_{3}^{[2]},v_1^{[2]},v_{1}^{[1]})$ & $(v_1 \xrightarrow{[1]} v_2 \xrightarrow{[2]} v_3 \xrightarrow{[2]} v_1)$ &   $l^{[1]}(v_1,v_2)+l^{[2]}(v_2,v_3)+l^{[2]}(v_3,v_1)$ \\ \hline
    $(u,v_{1}^{[1]},v_{2}^{[1]}, v_{2}^{[2]},v_{3}^{[2]},u)$ & $(u \rightarrow v_1 \xrightarrow{[1]} v_2 \xrightarrow{[2]} v_3 \xrightarrow{[2]} u)$ &   $l^{[1]}(v_1,v_2)+l^{[2]}(v_2,v_3)+l^{[2]}(v_3,u)$ \\ \hline \hline
  \end{tabular}
  \label{table:dicircuits}
\end{table}

Thus, we conclude that we can count the directed circuits in the labeled multi-graph $D$ to verify the existence of the potential function.
\end{proof}

\subsubsection*{\underline{GDoF Region Identification}} 
Now, operating on the potential digraphs, we are able to eliminate $r_k$'s in $\Pc'$ such that only $\{d_k^{[m]}\}_{k,m}$ remain.

According to Lemma \ref{lemma:potential-func}, a GDoF tuple $(\{d_k^{[m]}\}_{k,m})$ in $\Pc'$ is feasible if and only if there exists a potential function for the simple directed graph $D'$.
According to the potential theorem \cite[Th. 8.2]{Book-Comb-Opt}, the potential function exists if and only if each directed circuit in $D'$ has a non-negative sum-length. By Lemma \ref{lemma:digraph-equiv}, it suffices to impose that the directed circuits on the labeled multi-digraph $D$ are non-negative. In this way, we are able to identify $\Pc'$ without involving $r_k$'s.
Next, we divide the directed circuits into the following classes.
\begin{itemize}
\item Class I:  Directed circuits in the form of $(u \xrightarrow{} v_k \xrightarrow{[m]} u)$ for all $m\in [2]$ and $k \in [3]$.
\begin{subequations}
\begin{align}
\alpha_{kk}^{[1]} - \sum_{m=1}^{\pi_k(1)}d_k^{[m]} \ge 0 &\Leftrightarrow \sum_{m=1}^{\pi_k(1)}d_k^{[m]} \le \alpha_{kk}^{[1]}\\
\alpha_{kk}^{[2]} - \sum_{m=1}^{\pi_k(2)}d_k^{[m]} \ge 0 &\Leftrightarrow \sum_{m=1}^{\pi_k(2)}d_k^{[m]} \le \alpha_{kk}^{[2]}. \label{eq:individual}
\end{align}
\end{subequations}
\item Class II: Directed circuits in the form of $(v_k \xrightarrow{[m_1]} v_j \xrightarrow{[m_2]} v_k)$ for all $k \ne j \in [3]$ and $(m_1,m_2)\in\{(1,1),(1,2),(2,1),(2,2)\}$. For instance, when $(m_1,m_2) =(1,2)$, we have
\begin{subequations} \label{eq:3-user-potential-2-user}
\begin{align} 
&\alpha_{kk}^{[1]} - \sum_{m=1}^{\pi_k(1)}d_k^{[m]} - \alpha_{kj}^{[1]} + \alpha_{jj}^{[2]} - \sum_{m=1}^{\pi_j(2)}d_j^{[m]} - \alpha_{jk}^{[2]} \ge 0\\
& \qquad \Leftrightarrow \sum_{m=1}^{\pi_k(1)}d_k^{[m]} + \sum_{m=1}^{\pi_j(2)}d_j^{[m]} \le \alpha_{kk}^{[1]}+ \alpha_{jj}^{[2]}- \alpha_{kj}^{[1]} - \alpha_{jk}^{[2]}.
\end{align}
\end{subequations}
\item Class III: Directed circuits in the form of $(u \xrightarrow{} v_k \xrightarrow{[m_1]} v_j \xrightarrow{[m_2]} u)$ for all $k \ne j \in [3]$ and $(m_1,m_2)\in\{(1,1),(1,2),(2,1),(2,2)\}$. For instance, when $(m_1,m_2)=(1,2)$, we have
\begin{subequations}
\begin{align}
&\alpha_{kk}^{[1]} - \sum_{m=1}^{\pi_k(1)}d_k^{[m]} + \alpha_{jj}^{[2]} - \sum_{m=1}^{\pi_j(2)}d_j^{[m]}) - \alpha_{jk}^{[2]} \ge 0\\
& \qquad \Leftrightarrow \sum_{m=1}^{\pi_k(1)}d_k^{[m]} + \sum_{m=1}^{\pi_j(2)}d_j^{[m]} \le \alpha_{kk}^{[1]}+ \alpha_{jj}^{[2]} - \alpha_{jk}^{[2]},
\end{align}
\end{subequations}
which are implied by \eqref{eq:3-user-potential-2-user}, because $\alpha_{jk}^{[m]} \ge 0$ for all $j,k \in [3]$ and $m \in [2]$.
\item Class IV: Directed circuits in the form of $(v_k \xrightarrow{[m_1]} v_j \xrightarrow{[m_2]} v_i \xrightarrow{[m_3]} v_k)$ for all $(m_1,m_2,m_3) \in [2]^3$ and for either $(k,j,i)=(1,2,3)$ or $(k,j,i)=(1,3,2)$. For instance, when $(m_1,m_2,m_3)=(1,2,2)$, we have
\begin{subequations} \label{eq:3-user-potential-3-user}
\begin{align}
&\alpha_{kk}^{[1]} - \sum_{m=1}^{\pi_k(1)}d_k^{[m]} - \alpha_{kj}^{[1]} + \alpha_{jj}^{[2]} - \sum_{m=1}^{\pi_j(2)}d_j^{[m]} - \alpha_{ji}^{[2]} + \alpha_{ii}^{[2]} - \sum_{m=1}^{\pi_i(2)}d_i^{[m]} - \alpha_{ik}^{[2]} \ge 0 \\
& \qquad \Leftrightarrow \sum_{m=1}^{\pi_k(1)}d_k^{[m]} + \sum_{m=1}^{\pi_j(2)}d_j^{[m]} + \sum_{m=1}^{\pi_i(2)}d_i^{[m]} \le \alpha_{kk}^{[1]}+ \alpha_{jj}^{[2]} + \alpha_{ii}^{[2]}- \alpha_{kj}^{[1]} - \alpha_{ji}^{[2]} - \alpha_{ik}^{[2]}.
\end{align}
\end{subequations}
\item Class V: Directed circuits in the form of $(u \xrightarrow{[m_1]} v_k \xrightarrow{[m_2]} v_j \xrightarrow{[m_3]} v_i,u)$ for all $(m_1,m_2,m_3) \in [2]^3$ and for either $(k,j,i)=(1,2,3)$ or $(k,j,i)=(1,3,2)$. Similarly, the resulting sum GDoF inequalities are implied by those in \eqref{eq:3-user-potential-3-user}.
\end{itemize}

To sum up, after removing the redundant inequalities, we are left with \eqref{eq:individual} for all $k \in [3]$, \eqref{eq:3-user-potential-2-user} for all $(k,j) \in \{(1,2), (2,3), (1,3)\}$, and \eqref{eq:3-user-potential-3-user} for all $(k,j,i)\in\{(1,2,3), (1,3,2)\}$. A concise expression is as follows.
\begin{subequations} \label{eq:ex-3-user-2-state}
\begin{align}
  \sum_{m=1}^{\pi_k(m')}d_k^{[m]} &\le \alpha_{kk}^{[m']}, \quad \forall k \in [3], \; \forall m' \in [2]  \label{eq:3-user-individual}\\
  \sum_{m=1}^{\pi_k(m_1)}d_k^{[m]} + \sum_{m=1}^{\pi_j(m_2)}d_j^{[m]} &\le (\alpha_{kk}^{[m_1]}- \alpha_{kj}^{[m_1]}) + (\alpha_{jj}^{[m_2]} - \alpha_{jk}^{[m_2]}), \nonumber \\
  &\qquad\qquad\qquad \forall (k,j) \in \Pi_2, \; \forall (m_1,m_2) \in [2]^2\\
  \sum_{m=1}^{\pi_k(m_1)}d_k^{[m]} + \sum_{m=1}^{\pi_j(m_2)}d_j^{[m]} + \sum_{m=1}^{\pi_i(m_3)}d_i^{[m]} &\le (\alpha_{kk}^{[m_1]}- \alpha_{kj}^{[m_1]})+ (\alpha_{jj}^{[m_2]} - \alpha_{ji}^{[m_2]})+ (\alpha_{ii}^{[m_3]}  - \alpha_{ik}^{[m_3]}),\nonumber \\ &\qquad\qquad\qquad \forall (k,j,i) \in \Pi_3, \; \forall (m_1,m_2,m_3) \in [2]^3.
 \end{align}
 \end{subequations}
A more compact form of the last two sets of inequalities is
 \begin{align}
 \sum_{k=1}^{k'} \sum_{m=1}^{\pi_{i_k}(m_{i_k})}d_{i_k}^{[m]} &\le  \sum_{k=1}^{k'} (\alpha_{i_k i_k}^{[m_{i_k}]} - \alpha_{i_k i_{k+1}}^{[m_{i_k}]}), \label{3-user-sum}\\
&  \forall  (i_1, i_2, \dots, i_{k'}) \in \Pi_{k'}, \; \forall (m_{i_1}, m_{i_2}, \dots, m_{i_{k'}}) \in [2]^{k'}, \; \forall k' \in  \{2,3\}.
  \end{align}
It can be verified that, when $k'=2$, we have $(i_1,i_2) \in \Pi_2=\{(1,2), (2,3), (1,3)\}$ and $(m_{i_1},m_{i_2}) \in [2]^2= \{(1,1), (1,2), (2,1), (2,2)\}$, and thus the inequalities in \eqref{3-user-sum} correspond to those in Class II; when $k'=3$, we have $(i_1,i_2,i_3) \in \Pi_3=\{(1,2,3), (1,3,2)\}$ and $(m_{i_1},m_{i_2}, m_{i_3}) \in [2]^3=\{(1,1,1), (1,1,2), (1,2,1), (1,2,2), (2,1,1), (2,1,2), (2,2,1), (2,2,2)\}$, and thus the inequalities in \eqref{3-user-sum} correspond to those in Class IV.

Finally, note that the inequalities in (\ref{eq:ex-3-user-2-state}) match exactly those in $\mathcal{P}^*$. Therefore $\Pc' = \mathcal{P}^*$ and the achievability proof is complete.

{\subsection{The General Proof}\label{sec:main}
By simple layered superposition coding and opportunistic TIN decoding, the achievable GDoF value of the message $W_k^{[m]}$ via polyhedral TIN, with the $\max\{0,\cdot\}$ term ignored in \eqref{eq:GDoF}, is given by
\begin{align} \label{eq:general-poly-tin-gdof}
d_k^{[m]} = \min \Big\{r_k^{[m]} - r_k^{[m+1]}, \min_{m': \; \pi_k(m') \geq m} \big\{\alpha_{kk}^{[m']} + r_k^{[m]} - \max \{0, &\max_{i: i \neq k} (\alpha_{ki}^{[m']} + r_i^{[1]})\} \big \} \Big\}, \nonumber \\ 
&  \forall k \in [K], \; \forall m \in [M].
\end{align}
Thus, the polyhedral TIN achievable GDoF region $\Pc$ will be the set of GDoF tuples $(\{d_{k}^{[m]}\}_{k,m}) \in \RR_+^{MK}$, for which there exist $\{r_k^{[m]}\}_{k,m} \in \RR_-^{MK}$, such that all equations in \eqref{eq:general-poly-tin-gdof} are satisfied. In general, the polyhedral TIN region can only shrink the achievable GDoF region of TIN \cite{Geng_TIN}. We aim to show that, when the TIN optimality condition \eqref{eq:tin-optimal} is satisfied, polyhedral TIN incurs no loss, and achieves the optimal GDoF region $\Pc^*$ in \eqref{eq:GDoF-region}.

{In what follows, we first impose a constraint on \eqref{eq:general-poly-tin-gdof} to construct an achievable GDoF region $\Pc' \subseteq \Pc$, and then by identifying $\Pc'$ and showing that it is the same as $\Pc^*$, we complete the achievability proof. %

\subsubsection{Constructing $\Pc' \subseteq \Pc$}
By imposing the following constraint
\begin{align} \label{eq:additional-cont-general}
d_k^{[m]} &= r_k^{[m]} - r_k^{[m+1]}, \quad \forall k \in [K], \; m \in [M],
\end{align}
\eqref{eq:general-poly-tin-gdof} reduces to
\begin{align} 
d_k^{[m]} &\le \min_{m': \; \pi_k(m') \geq m} \big\{\alpha_{kk}^{[m']} + r_k^{[m]} - \max \{0, \max_{i: i \neq k} (\alpha_{ki}^{[m']} + r_i^{[1]})\} \big \}, \quad \forall k \in [K], \; \forall m \in [M],
\end{align}
which further expands to
\begin{align} 
d_k^{[m]} &\le  \alpha_{kk}^{[m']} + r_k^{[m]} - \max \{0, \max_{i: i \neq k} (\alpha_{ki}^{[m']} + r_i^{[1]})\}, \\
&= \alpha_{kk}^{[m']} +  r_k^{[1]} - \sum_{m''=1}^{m-1} d_k^{[m'']} - \max \{0, \max_{i: i \neq k} (\alpha_{ki}^{[m']} + r_i^{[1]})\}, \nonumber \\
&  \qquad \qquad \qquad \qquad \forall m' \in [M], \; \st \; \pi_k(m') \geq m, \; \forall k \in [K], \; \forall m \in [M],
 \label{eq:general-poly-tin-gdof-relaxed}
\end{align}
due to
\begin{align}
r_k^{[m]} = r_k^{[1]} - \sum_{m''=1}^{m-1} d_k^{[m'']}.
\end{align}
Rearranging \eqref{eq:general-poly-tin-gdof-relaxed}, we have 
\begin{align}
 \sum_{m''=1}^{m} d_k^{[m'']} &\le \alpha_{kk}^{[m']} +  r_k^{[1]} - \max \{0, \max_{i: i \neq k} (\alpha_{ki}^{[m']} + r_i^{[1]})\}, \nonumber \\
 & \qquad  \qquad \qquad \forall m \le \pi_k(m'), \; \forall k \in [K], \; \forall m' \in [M].
\end{align}
With respect to $m$, the inequality with $m = \pi_k(m')$ is the dominant one and implies others with $m < \pi_k(m')$, because of the non-negativity of $\{d_{k}^{[m]}\}_{k,m}$.

Hence, we have constructed $\Pc'$ with respect to $(\{d_{k}^{[m]}\}_{k,m}) \in \RR_+^{MK}$ (for some properly chosen parameters $(\{r_{k}^{[m]}\}_{k,m})$), defined by the following inequalities.
\begin{subequations}
\begin{align} 
r_{k}^{[m]} &\le 0, \quad \forall k \in [K], \; \forall m \in [M]\\
d_{k}^{[m]} &\ge 0, \quad \forall k \in [K], \; \forall m \in [M]\\
d_k^{[m]} &= r_k^{[m]} - r_k^{[m+1]}, \quad \forall k \in [K], \; \forall m \in [M] \\
\sum_{m=1}^{\pi_k(m')} d_k^{[m]} &\le \alpha_{kk}^{[m']} +  r_k^{[1]} - \max \{0, \max_{i: i \neq k} (\alpha_{ki}^{[m']} + r_i^{[1]})\}, \quad \forall k \in [K], \; \forall m' \in [M],
\end{align}
\end{subequations}
where the additional constraint \eqref{eq:additional-cont-general} makes $\Pc'$ no larger than $\Pc$, i.e., $\Pc' \subseteq \Pc$.  

\subsubsection{Proof of $\Pc' = \mathcal{P}^*$}
Next, we eliminate $\{r_k^{[m]}\}_{k,m}$ in $\Pc'$ and show that it becomes $\mathcal{P}^*$.

Due to the imposed power relation in \eqref{eq:additional-cont-general}, $\{r_k^{[m]}, m \ge 2\}_k$ can be recursively computed and we only need to focus on the existence of $\{r_k^{[1]}\}$ (for the basic messages) with regard to
\begin{align} 
 \sum_{m=1}^{\pi_k(m')} d_k^{[m]} \le \alpha_{kk}^{[m']} + r_k - \max \{0, \max_{j: j \neq k} (\alpha_{kj}^{[m']} + r_j)\}, \quad \forall k \in [K], \; \forall m' \in [M]
\end{align}
where we set $r_k^{[1]}=r_k$ for the sake of notational brevity.
}

For a given GDoF tuple $(\{d_k^{[m]}\}_{k,m}) \in \RR_+^{MK}$, it is feasible in $\Pc'$ if and only if there exist $\{r_k\}_k$'s satisfying
 \begin{subequations}
\begin{align}
r_k &\le 0,\\
r_k &\ge \sum_{m=1}^{\pi_k(m')} d_k^{[m]} - \alpha_{kk}^{[m']}, \quad \forall m' \in [M] \\
r_k -r_j &\ge \sum_{m=1}^{\pi_k(m')} d_k^{[m]} - \alpha_{kk}^{[m']} + \alpha_{kj}^{[m']},  \quad  \forall m' \in [M], \; \forall j \neq k.
\end{align} 
\end{subequations}

Similarly to the 3-user example, to verify the existence of $\{r_k\}_k$, we construct a potential digraph to ensure the existence of a valid potential function. 
For the simplicity of presentation, we only focus on the labeled multi-digraph. For the general $K$-user interference channel with $M$ states, the labeled multi-digraph $D=(V,A)$ is such that
\begin{align}
V &= \{u, v_1, v_2, \dots, v_K\} \\
A &=\{(u,v,m',l): u,v \in V, m' \in [M], l \in \RR\}
\end{align}
where $m'$ specifies the label of an arc and $l=l^{[m']}(u,v)$ is the length between $u,v \in V$ with label $m'$. The arc set consists of three parts $A = A_1 \cup A_2 \cup A_3$, where
\begin{subequations}
\begin{align}
A_1 &= \{(u,v_k,,l(u,v_k)): k \in [K]\} \\
A_2 &= \{(v_k,u,m',l^{[m']}(v_k,u)): m' \in [M], k \in [K]\}\\
A_3 &= \{(v_k,v_j,m',l^{[m']}(v_k,v_j)): m' \in [M], k, j \in [K], k \ne j\} 
\end{align}
\end{subequations}
with a length $l^{[m']}(a,b)$ assigned to every single arc $(a,b) \in \Ac$ as follows:
\begin{subequations}
\begin{align} 
l(u,v_k) &= 0\\
l^{[m']}(v_k,u) &= \alpha_{kk}^{[m']} - \sum_{m=1}^{\pi_k(m')} d_k^{[m]}, \quad \forall m' \in [M]\\
l^{[m']}(v_k,v_j) &= \alpha_{kk}^{[m']} - \sum_{m=1}^{\pi_k(m')} d_k^{[m]} - \alpha_{kj}^{[m']} \quad \forall m' \in [M]
\end{align}
\end{subequations}

According to the potential theorem \cite[Th. 8.2]{Book-Comb-Opt}, %
by imposing that the lengths of the shortest directed circuits in the labeled multi-digraph $D$ are non-negative, the existence of a potential function is guaranteed. The imposed non-negativity lends itself to the identification of $\Pc'$ without involving $\{r_k\}_k$.
\begin{itemize}
\item Class I: Directed circuits in the form of $(u \xrightarrow{} v_k \xrightarrow{[m']} u)$ for all $m' \in [M]$ and $k \in [K]$.
\begin{align} \label{eq:individual-general}
\alpha_{kk}^{[m']} - \sum_{m=1}^{\pi_k(m')} d_k^{[m]} \ge 0 &\Leftrightarrow \sum_{m=1}^{\pi_k(m')} d_k^{[m]} \le \alpha_{kk}^{[m']}.
\end{align}

\item Class II: For all $k' \in [K] \backslash \{1\}$ and $(i_1, i_2, \dots, i_{k'}) \in \Pi_{k'}$, directed circuits in the form of $(v_{i_1} \xrightarrow{[m_{i_1}]} v_{i_2} \xrightarrow{[m_{i_2}]} \dots \xrightarrow{[m_{i_{k'-1}}]} v_{i_k'} \xrightarrow{[m_{i_{k'}}]} v_{i_1})$ for all $(m_{i_1},m_{i_2}, \dots, m_{i_{k'}})\in [M]^{k'}$.
\begin{subequations} \label{eq:all-v-cycle-general}
\begin{align} 
\MoveEqLeft \sum_{k=1}^{k'} \left( \alpha_{i_k i_k}^{[m_{i_k}]} - \sum_{m=1}^{\pi_{i_k}(m_{i_k})} d_k^{[m]} - \alpha_{i_k i_{k+1}}^{[m_{i_k}]} \right) \ge 0\\
&  \Leftrightarrow \sum_{k=1}^{k'} \sum_{m=1}^{\pi_{i_k}(m_{i_k})} d_k^{[m]}  \le \sum_{k=1}^{k'} \left( \alpha_{i_k i_k}^{[m_{i_k}]} - \alpha_{i_k i_{k+1}}^{[m_{i_k}]} \right).
\end{align}
\end{subequations}
\item Class III: For all $k' \in [K] \backslash \{1\}$ and $(i_1, i_2, \dots, i_{k'}) \in \Pi_{k'}$, directed circuits in the form of $(u \xrightarrow{} v_{i_1} \xrightarrow{[m_{i_1}]} v_{i_2} \xrightarrow{[m_{i_2}]} \dots \xrightarrow{[m_{i_{k'-1}}]} v_{i_k'} \xrightarrow{[m_{i_{k'}}]}  u)$ for all $(m_{i_1},m_{i_2}, \dots, m_{i_{k'}})\in [M]^{k'}$.
\begin{subequations}
\begin{align} 
\MoveEqLeft \sum_{k=1}^{k'-1} \left( \alpha_{i_k i_k}^{[m_{i_k}]} - \sum_{m=1}^{\pi_{i_k}(m_{i_k})} d_k^{[m]} - \alpha_{i_k i_{k+1}}^{[m_{i_k}]} \right) + \left( \alpha_{i_{k'} i_{k'}}^{[m_{i_{k'}}]} - \sum_{m=1}^{\pi_{i_{k'}}(m_{i_{k'}})} d_{k'}^{[m]} \right) \ge 0\\
&  \Leftrightarrow \sum_{k=1}^{k'} \sum_{m=1}^{\pi_{i_k}(m_{i_k})} d_k^{[m]}  \le \sum_{k=1}^{k'} \left( \alpha_{i_k i_k}^{[m_{i_k}]} - \alpha_{i_k i_{k+1}}^{[m_{i_k}]} \right) + \alpha_{i_{k'} i_{1}}^{[m_{i_{k'}}]}
\end{align}
\end{subequations}
which is implied by \eqref{eq:all-v-cycle-general}.
\end{itemize}
}
It is not hard to verify that apart from the circuits above mentioned, there are no other shortest directed circuits. By far, we have simplified $\Pc'$ such that it is represented with respect only to $\{d_k^{[m]}\}_{k,m}$. %
Collecting the inequalities of \eqref{eq:individual-general} and \eqref{eq:all-v-cycle-general}, we find that $\Pc' = \Pc^*$, when the TIN optimality condition \eqref{eq:tin-optimal} is satisfied. The achievability proof is thus complete.

\section{Converse}
\label{sec:converse}
{For the converse, instead of starting from Fano's inequality and upper-bounding the sum rate, we cast our problem to a set of regular interference channels for which the optimal GDoF regions under TIN-optimality conditions have been characterized in \cite{Geng_TIN}. In this way, we directly collect the sum GDoF constraints therein to form our GDoF region outer bound. In doing so, the converse proof can be significantly simplified.}

We use the set of channel coefficients to indicate different states, i.e.,
\begin{align}
\text{State $m$: } \Hc^{[m]} \defeq \left\{ (\{h_{ki}^{[m]}\}_{i,k}) \right\} \subseteq \CC^{K^2}
\end{align}
where at State $m$, Receiver $k$ wishes to recover messages $\{W_k^{[l]}\}_{l=1}^{\pi_k(m)}$.
Further, we define $\tilde{\mv}=(m_{i_1},\dots,m_{i_K})$ and introduce
\begin{align}
\text{State $\tilde{\mv}$: } \tilde{\Hc}^{[\tilde{\mv}]} \defeq \left\{ \Hc_{1}^{[m_{i_1}]}, \Hc_{2}^{[m_{i_2}]}, \dots, \Hc_{K}^{[m_{i_K}]} \right \}  \subseteq \CC^{K^2}
\end{align}
where $\tilde{\mv} \in [M]^K$, $\Hc_{k}^{[m_{i_k}]}\defeq \left\{ (\{h_{ki}^{[m_{i_k}]}\}_{i}) \right\}$ and at State $\tilde{\mv}$, Receiver $k$ wishes to recover messages $\{W_k^{[l]}\}_{l=1}^{\pi_k(m_{i_k})}$. Apparently, $ \Hc^{[m]}$ is a realization of $ \tilde{\Hc}^{[\tilde{\mv}]}$ when $\tilde{\mv}=(m, m, \dots,m)$.

According to the construction of the states, besides $M$ original states, we also introduce $M_K-M$ auxiliary states. We make the following statement.
\begin{lemma} \label{lemma:aux-state-decode}
Any message set in the $M$-state Gaussian interference channel (GIC) defined by $\{\Hc^{[m]}\}_m$ can be decoded if and only if the same message set can be decoded in the $M^K$-state GIC defined by $\{\tilde{\Hc}^{[\tilde{\mv}]}, \tilde{\mv} \in [M]^K \}$.
\end{lemma}
\begin{proof}
The ``if'' part is readily obtained, because the states $\{\Hc^{[m]}\}_m$ are a subset of the states $\{\tilde{\Hc}^{[\tilde{\mv}]}, \tilde{\mv} \in [M]^K \}$. Thus, the messages decoded in the latter can be decoded in the former.

For the ``only if'' part, we need to show that, if the messages are decodable in the $M$-state GIC, these messages are also decodable in all the auxiliary states. 
Consider State $m' \in [M]$ in the $M$-state GIC such that the message set $\{{W}_k^{[m]}\}_{m=1}^{\pi_k(m')}$ at every receiver $k$ can be decoded. Then the average probability of error satisfies
\begin{align}
\lim_{n \to \infty} \Pr\left(\{{W}_k^{[1:\pi_k(m')]}\}_k \neq \{\hat{W}_k^{[1:\pi_k(m')]}\}_k \right) = 0, \label{eq:pe1}
\quad \text{for all $m' \in [M]$}
\end{align}
given the encoding and decoding mappings $\tilde{X}_i(t) = f_i(\{W_i^{[m]}\}_m)$ for all $i \in [K]$ and $\{\hat{W}_k^{[m]}\}_{m=1}^{\pi_k(m')} = g_k^{[m']} (Y_{k,m'}^n)$ for all $k \in [K]$.

Without loss of generality, we focus on a specific auxiliary state $\tilde{\mv}=(m_{i_1},\dots,m_{i_K})$ in the $M^K$-state GIC, where the input-output relation is as follows:
\begin{align}
\bar{Y}_k^{[m_{i_k}]}(t) = \sum_{i=1}^{K}{h}_{ki}^{[m_{i_k}]} \bar{X}_i(t) + \bar{Z}_k^{[m_{i_k}]}(t).
\end{align}
We impose that $\bar{X}_i(t)=\tilde{X}_i(t)=f_i(\{W_i^{[m]}\}_m)$ for all $i \in [K]$, i.e., the input $\bar{X}_i(t)$ in the $M^K$-state GIC has the same encoding mapping applied at each transmitter as used in the $M$-state GIC. Thus, the received signal for Receiver $k$ can be rewritten as
\begin{align}
\bar{Y}_k^{[m_{i_k}]}(t) &= \sum_{i=1}^{K}{h}_{ki}^{[m_{i_k}]} \tilde{X}_i(t) + \bar{Z}_k^{[m_{i_k}]}(t)\\
&\sim \sum_{i=1}^{K}{h}_{ki}^{[m_{i_k}]} \tilde{X}_i(t) + {Z}_k^{[m_{i_k}]}(t) = {Y}_k^{[m_{i_k}]}(t)
\end{align}
where $A \sim B$ means that $A$ and $B$ are statistically equivalent. So the received signal in the $M^K$-state GIC is statistically equivalent to that in the $M$-state GIC. Applying the same decoding mapping $g_k^{[m_{i_k}]}$ as that in State $m_{i_k}$ of the $M$-state GIC, we have
\begin{align}
\MoveEqLeft \lim_{n \to \infty} \Pr\left(\{\{{W}_k^{[m]}\}_{m=1}^{\pi_{k}(m_{i_k})}\}_k \neq \{\{\hat{W}_k^{[m]}\}_{m=1}^{\pi_{k}(m_{i_k})}\}_k \right)\\
 & =\lim_{n \to \infty}  \Pr\left(\{\{{W}_k^{[m]}\}_k\}_{m=1}^{\pi_{k}(m_{i_k})} \neq \{\{\hat{W}_k^{[m]}\}_k \}_{m=1}^{\pi_{k}(m_{i_k})} \right)\\
 & \le \lim_{n \to \infty}  \Pr\left(\{\{{W}_k^{[m]}\}_k\}_{m=1}^{\max_k \pi_{k}(m_{i_k})} \neq \{\{\hat{W}_k^{[m]}\}_k \}_{m=1}^{\max_k \pi_{k}(m_{i_k})} \right)\\
 & = \lim_{n \to \infty}  \Pr\left(\{\{{W}_k^{[m]}\}_{m=1}^{\max_k \pi_{k}(m_{i_k})}\}_k \neq \{\{\hat{W}_k^{[m]} \}_{m=1}^{\max_k \pi_{k}(m_{i_k})}\}_k \right)\\
 &\overset{(\ref{eq:pe1})}{=} 0.
\end{align}
Therefore the messages can indeed be decoded at the auxiliary states. This completes the proof.
\end{proof}

Whether the messages can be decoded at a receiver is determined by the marginal distribution associated to this receiver if there is no receiver cooperation. Thus the same message set can be decoded in both the $M$-state GIC and the $M^K$-state GIC as the receivers in the two networks see the same marginal channel transition probabilities. Similar statements have been used extensively in network information theory literature (e.g.,  \cite[Lemma 5.1]{NIT}, \cite[Proposition 2]{RPV09}).

By Lemma \ref{lemma:aux-state-decode}, we conclude that the achievable rate tuple $(\{R_{k}^{[m]}\}_{m,k})$ in the $M$-state GIC should satisfy the sum rate constraints in the $M^K$-state GIC. Given a state $\tilde{\mv}=(m_{i_1},\dots,m_{i_K})$, we treat the set of messages $\{W_{i_k}^{[m]}\}_{m=1}^{\pi_{i_k}(m_{i_k})}$ as a single virtual message $\tilde{W}_{i_k}$. Let $d_{i_k}$ be the GDoF of $\tilde{W}_{i_k}$. As such, we have $d_{i_k}:=\sum_{m=1}^{\pi_{i_k}(m_{i_k})}  d_{i_k}^{[m]}$.
Such a state is a regular interference channel with messages $\{\tilde{W}_{i_k}\}_{k}$ and here the TIN optimality condition is satisfied (refer to \eqref{eq:tin-optimal}), so by Theorem 1 of \cite{Geng_TIN} the GDoF tuple $(\{d_{i_k}\}_{k})$ should satisfy
\begin{subequations} \label{eq:outer-bound}
\begin{align}
 d_{i_k} &\le \alpha_{i_ki_k}^{[m_{i_k}]} , \quad \forall k \label{eq:ind-sum-gdof-1state}\\
 \sum_{k=1}^{k'}  d_{i_k} &\le  \sum_{k=1}^{k'} (\alpha_{i_k i_k}^{[m_{i_k}]} - \alpha_{i_k i_{k+1}}^{[m_{i_k}]}), \label{eq:sum-gdof-1state} \\
  &\forall  (i_1, i_2, \dots, i_{k'}) \in \Pi_{k'}, \forall k' \in [K] \backslash \{1\}.  
 \end{align}
 \end{subequations}
 
Collecting all inequalities for all possible states $(m_{i_1}, m_{i_2}, \dots, m_{i_{k'}}) \in [M]^{k'}$, we have that the constraints of \eqref{eq:ind-sum-gdof-1state} for all possible states are equivalent to (\ref{eq:e1}). For a specific state $(m_{i_1}, m_{i_2}, \dots, m_{i_{k'}}) \in [M]^{k'}$, the constraints \eqref{eq:sum-gdof-1state} match exactly those in (\ref{eq:e2}). The outer bound proof is thus complete.

\section{Conclusion}
Motivated by the need to communicate with a higher rate when channels are in better conditions (i.e., opportunistic communications), we consider a $K$-user interference network with multiple channel states and degraded message sets, where each transmitter has a set of messages (ordered by their priorities) and each receiver will decode a number of messages up to a pre-determined threshold on the message order, depending on the channel state. For this channel with states, we show that if each sub-network (comprised of receivers from possibly distinct states) satisfies a TIN-optimality condition, then simple layered superposition encoding and successive cancelation based opportunistic TIN decoding achieves the entire GDoF region, for all possible decoding thresholds at each receiver.


\begin{thebibliography}{10}
\providecommand{\url}[1]{#1}
\csname url@samestyle\endcsname
\providecommand{\newblock}{\relax}
\providecommand{\bibinfo}[2]{#2}
\providecommand{\BIBentrySTDinterwordspacing}{\spaceskip=0pt\relax}
\providecommand{\BIBentryALTinterwordstretchfactor}{4}
\providecommand{\BIBentryALTinterwordspacing}{\spaceskip=\fontdimen2\font plus
\BIBentryALTinterwordstretchfactor\fontdimen3\font minus
  \fontdimen4\font\relax}
\providecommand{\BIBforeignlanguage}[2]{{%
\expandafter\ifx\csname l@#1\endcsname\relax
\typeout{** WARNING: IEEEtran.bst: No hyphenation pattern has been}%
\typeout{** loaded for the language `#1'. Using the pattern for}%
\typeout{** the default language instead.}%
\else
\language=\csname l@#1\endcsname
\fi
#2}}
\providecommand{\BIBdecl}{\relax}
\BIBdecl

\bibitem{wc_book_tse}
D.~Tse and P.~Viswanath, \emph{Fundamentals of Wireless Communication}.\hskip
  1em plus 0.5em minus 0.4em\relax Cambridge University Press, 2005.

\bibitem{Pramod}
P.~Viswanath, D.~Tse, and R.~Laroia, ``Opportunistic beamforming using dumb
  antennas,'' \emph{IEEE Transactions on Information Theory}, vol.~48, no.~6, pp.
  1277--1294, June 2002.

\bibitem{devroye2008cognitive}
N.~Devroye, M.~Vu, and V.~Tarokh, ``Cognitive radio networks,'' \emph{IEEE
  Signal Processing Magazine}, vol.~25, no.~6, 2008.

\bibitem{goldsmith2009breaking}
A.~Goldsmith, S.~A. Jafar, I.~Maric, and S.~Srinivasa, ``Breaking spectrum
  gridlock with cognitive radios: An information theoretic perspective,''
  \emph{Proceeding of IEEE}, vol.~97, no.~5, pp. 894--914, 2009.

\bibitem{Diggavi_Tse}
S.~N. Diggavi and D.~Tse, ``Fundamental limits of diversity-embedded codes over
  fading channels,'' in \emph{IEEE International Symposium on Information Theory (ISIT)}, 2005, pp. 510--514.

\bibitem{Khude_Prabhakaran_Viswanath}
N.~Khude, V.~Prabhakaran, and P.~Viswanath, ``Harnessing bursty interference,''
  in \emph{IEEE Information Theory Workshop (ITW)}, 2009, pp. 13--16.

\bibitem{Khude_Prabhakaran_Viswanath2}
------, ``Opportunistic interference management,'' in \emph{IEEE International Symposium on Information Theory (ISIT)}, 2009, pp. 2076--2080.

\bibitem{Geng_TIN}
C.~Geng, N.~Naderializadeh, S.~Avestimehr, and S.~Jafar, ``{On the Optimality
  of Treating Interference as Noise},'' \emph{IEEE Transactions on Information Theory}, vol.~61, no.~4, pp. 1753 -- 1767, April 2015.

\bibitem{TIN-X}
C.~Geng, H.~Sun, and S.~A. Jafar, ``On the optimality of treating interference
  as noise: General message sets.'' \emph{IEEE Transactions on Information Theory},
  vol.~61, no.~7, pp. 3722--3736, 2015.

\bibitem{Sun_Jafar_ParallelTIN}
H.~Sun and S.~A. Jafar, ``On the optimality of treating interference as noise
  for $ k $-user parallel gaussian interference networks,'' \emph{IEEE Transactions on Information Theory}, vol.~62, no.~4, pp. 1911--1930, 2016.

\bibitem{Sun_Jafar_ParallelTINRegion}
H.~Sun and S.~Jafar, ``On the separability of gdof region for parallel gaussian
  tin optimal interference networks,'' in \emph{IEEE International Symposium on
  Information Theory (ISIT)}, 2015.

\bibitem{TIN-Compound}
C.~Geng and S.~A. Jafar, ``On the optimality of treating interference as noise:
  Compound interference networks,'' \emph{IEEE Transactions on Information
  Theory}, vol.~62, no.~8, pp. 4630--4653, 2016.

\bibitem{Aydin-TIN}
S.~Gherekhloo, A.~Chaaban, C.~Di, and A.~Sezgin, ``(sub-)optimality of treating
  interference as noise in the cellular uplink with weak interference,''
  \emph{IEEE Transactions on Information Theory}, vol.~62, no.~1, pp. 322--356,
  Jan 2016.

\bibitem{TIN-IMAC}
H.~Joudeh and B.~Clerckx, ``On the optimality of treating interference as noise
  for interfering multiple access channels,'' \emph{arXiv preprint
  arXiv:1805.04773}, 2018.

\bibitem{Yi-TIN}
X.~Yi and G.~Caire, ``Optimality of treating interference as noise: A
  combinatorial perspective,'' \emph{IEEE Transactions on Information Theory},
  vol.~62, no.~8, pp. 4654--4673, 2016.

\bibitem{ITLinQ}
N.~Naderializadeh and A.~S. Avestimehr, ``{ITLinQ}: A new approach for spectrum
  sharing in device-to-device communication systems,'' \emph{IEEE Journal on
  Selected Areas in Communications}, vol.~32, no.~6, pp. 1139--1151, June 2014.

\bibitem{ITLinQ+}
X.~Yi and G.~Caire, ``{ITLinQ+}: An improved spectrum sharing mechanism for
  device-to-device communications,'' in \emph{49th Asilomar Conference on
  Signals, Systems and Computers}, Nov 2015, pp. 1310--1314.

\bibitem{Book-Comb-Opt}
A.~Schrijver, \emph{Combinatorial Optimization: Polyhedra and
  Efficiency}.\hskip 1em plus 0.5em minus 0.4em\relax Springer Science \&amp;
  Business Media, 2003, vol.~24.

\bibitem{NIT}
A.~El~Gamal and Y.-H. Kim, \emph{Network Information Theory}.\hskip 1em plus
  0.5em minus 0.4em\relax Cambridge University Press, 2011.

\bibitem{RPV09}
A.~Raja, V.~M. Prabhakaran, and P.~Viswanath, ``The two-user compound
  interference channel,'' \emph{IEEE Transactions on Information Theory},
  vol.~55, no.~11, pp. 5100--5120, 2009.

\end{thebibliography}
\end{document}